\pdfoutput=1
\documentclass[
aps,
prx,
reprint,
superscriptaddress,
longbibliography]{revtex4-2}

\usepackage{graphicx}

\usepackage{amsmath,amssymb,amsfonts, ascmac,mathtools,bm,braket,stmaryrd}
\usepackage{enumitem}
\usepackage{comment}
\usepackage{derivative}

\usepackage{amsthm}

\theoremstyle{definition}
\newtheorem{thm}{Theorem}
\newtheorem{prop}[thm]{Proposition}
\newtheorem{lem}[thm]{Lemma}

\newcommand{\Hilb}{\mathcal{H}}
\newcommand{\bound}{\mathcal{B}}
\newcommand{\id}{\mathrm{id}}
\newcommand{\im}{\mathrm{i}}
\newcommand{\tbound}{\tilde{\mathcal{B}}}
\newcommand{\tot}{\mathrm{tot}}

\newcommand{\figurecentering}[1]{\vcenter{\hbox{#1}}}

\DeclareMathOperator{\size}{\mathrm{size}}
\DeclareMathOperator{\len}{\mathrm{len}}

\DeclareMathOperator{\Tr}{\mathrm{Tr}}
\DeclareMathOperator{\Span}{\mathrm{Span}}

\mathchardef\standardl=\mathcode`l
\mathcode`l=\ell
\newcommand{\deactivatel}{\mathcode`l=\standardl}
\makeatletter
\edef\operator@font{\operator@font\noexpand\deactivatel}
\makeatother

\usepackage{hyperref}
\hypersetup{colorlinks=true, citecolor=magenta, linkcolor=blue, urlcolor=cyan}

\renewcommand\paragraph[1]{%
\par\emph{#1---}\kern2pt\relax\ignorespaces}

\begin{document}

\title{Rigorous Test for Quantum Integrability and Nonintegrability}

\author{Akihiro Hokkyo}
\email{hokkyo@cat.phys.s.u-tokyo.ac.jp}
\affiliation{Department of Physics, Graduate School of Science, The University of Tokyo, 7-3-1 Hongo, Bunkyo, Tokyo, 113-8654, Japan}

\begin{abstract}
The integrability of a quantum many-body system, 
which is characterized by the presence or absence of local conserved quantities, 
drastically impacts the dynamics of isolated systems, 
including thermalization. 
Nevertheless, a rigorous and comprehensive method for determining integrability or nonintegrability has remained elusive.
In this paper, we address this challenge by introducing rigorously provable tests for integrability and nonintegrability of quantum spin systems with finite-range interactions. 
Our results significantly simplify existing proofs of nonintegrability,
such as those for
the $S=1/2$ Heisenberg chain with nearest-and next-nearest-neighbor interactions, 
the $S=1$  bilinear-biquadratic chain
and the $S=1/2$ XYZ model in two or higher
dimensions.
Moreover, our results also yield the first proof of nonintegrability for models such as the $S=1/2$ Heisenberg chain with a non-uniform magnetic field, 
the $S=1/2$ XYZ model on the triangular lattice, 
and the general spin XYZ model.
This work also offers a partial resolution to the long-standing conjecture that integrability is governed by the existence of local conserved quantities with small support. 
Our framework ensures that the nonintegrability of one-dimensional spin systems with translational symmetry can be verified algorithmically, independently of system size.
\end{abstract}

\maketitle
\section{Introduction}
A central challenge in quantum many-body physics is to understand the equilibrium states and dynamics of strongly interacting systems. Rigorous methods sometimes offer deep insights into these problems. One such approach is to establish physically meaningful results for broad classes of systems, exemplified by 
the Hohenberg--Mermin--Wagner theorem~\cite{hohenbergExistenceLongRangeOrder1967,merminAbsenceFerromagnetismAntiferromagnetism1966} and the Lieb--Schultz--Mattis theorem~\cite{liebTwoSolubleModels1961}. %
Another is the use of exactly solvable models.
The Bethe ansatz method, 
originating from Bethe's exact solution of the Heisenberg chain~\cite{betheZurTheorieMetalle1931}, 
enables the analytical calculation of quantities that are otherwise inaccessible in general systems, 
such as energy eigenvalues, eigenstates, and correlation functions. 
Furthermore, the algebraic Bethe ansatz method, or the quantum inverse scattering method~\cite{korepinQuantumInverseScattering1993}, 
has revealed 
that behind such solvability lies integrability, i.e., the existence of an extensive number of local conserved quantities.

However, when considering the dynamics of isolated quantum many-body systems, integrable systems can exhibit anomalous behavior.
A notable example is their deviation from typical thermalization behavior~\cite{moriThermalizationPrethermalizationIsolated2018}. 
Recent experiments have shown thermalization, 
i.e., the relaxation of local observables to thermal equilibrium
in well-isolated quantum systems~\cite{kaufmanQuantumThermalizationEntanglement2016}. 
The eigenstate thermalization hypothesis (ETH)~\cite{deutschQuantumStatisticalMechanics1991,srednickiChaosQuantumThermalization1994,rigolThermalizationItsMechanism2008}, 
which posits that energy eigenstates are locally indistinguishable from thermal equilibrium states, is widely accepted as the underlying mechanism. 
While the ETH has been numerically demonstrated to hold for various nonintegrable systems~\cite{kimTestingWhetherAll2014,garrisonDoesSingleEigenstate2018}, 
it cannot be applied to integrable systems due to the presence of local conserved quantities. 
Some experiments have also shown that integrable systems exhibit unusual relaxation dynamics~\cite{kinoshitaQuantumNewtonsCradle2006}. 
Moreover, the existence of local conserved quantities is known to give rise to nontrivial inequalities in linear response~\cite{mazurNonergodicityPhaseFunctions1969,suzukiErgodicityConstantsMotion1971}. 
Thus, significant differences exist between the universal behavior observed in nonintegrable systems and the anomalous behavior in integrable systems. 
Consequently, in addressing the fundamental question of statistical mechanics---how the empirical laws of macroscopic systems emerge from quantum mechanics---it is crucial to determine whether the system is integrable or nonintegrable.

While integrability has traditionally been a focus of rigorous research, 
nonintegrability has long remained outside its scope. 
However, recent studies have introduced methods for rigorously proving nonintegrability 
--or, more precisely, the absence of local conserved quantities~\cite{shiraishiProofAbsenceLocal2019,chibaMixfield2024,shiraishiAbsenceLocalConserved2024,chibaExactThermalEigenstates2024,parkGraphTheoreticalProof2024,yamaguchiCompleteClassificationIntegrability2024,yamaguchiProofAbsenceLocal2024,parkProofNonintegrabilitySpin12024,hokkyoProofAbsenceLocal2025,shiraishi2dimXY2024,chiba2dimIsing2024,shiraishiCompleteClassificationIntegrability2025}. 
The initial work focused on the spin-1/2 XYZ-h chain~\cite{shiraishiProofAbsenceLocal2019}, 
but subsequent studies~\cite{chibaMixfield2024,yamaguchiCompleteClassificationIntegrability2024,yamaguchiProofAbsenceLocal2024} classified the nonintegrability of general translationally invariant spin-1/2 chains with nearest-neighbor symmetric interactions. 
This research has expanded to include spin-1/2 chains with next-nearest-neighbor interactions~\cite{shiraishiAbsenceLocalConserved2024,parkGraphTheoreticalProof2024,shiraishiCompleteClassificationIntegrability2025}, 
spin-1 chains~\cite{parkProofNonintegrabilitySpin12024,hokkyoProofAbsenceLocal2025}, 
and higher-dimensional systems~\cite{shiraishi2dimXY2024,chiba2dimIsing2024}.
However, existing methods remain heuristic in nature, 
often relying on a form of ``craftsmanship'' to demonstrate 
the absence of $k$-local conserved quantities 
for each size $k$ of the operator's support.
Although all previous proofs of nonintegrability have employed similar methods, 
the underlying unifying structure remains largely unexplored. 
This lack of a unified framework also poses a significant challenge for computational approaches, 
as the dimension of the space of $k$-local operators grows exponentially with $k$. 
Since other numerical tests for nonintegrability, such as the investigation of level spacing statistics~\cite{bohigasCharacterizationChaoticQuantum1984,brezinSpectralFormFactor1997}, 
are inherently subject to finite-size effects, 
it is of fundamental importance to determine whether the rigorous proof method described above can be formulated as an efficient algorithm.

Efficient methods for determining integrability have long been of interest as well 
in the study of integrable systems~\cite{kulishQuantumSpectralTransform1982,dolanConservedChargesSelfduality1982,fuchssteinerComputerAlgorithmsDetection1988,kennedySolutionsYangBaxterEquation1992,grabowskiIntegrabilityTestSpin1995,gomborIntegrableSpinChains2021}.
It has been conjectured that the existence of $3$-local conserved quantities is a necessary condition for integrability in one-dimensional systems with translational symmetry and nearest-neighbor interactions~\cite{grabowskiIntegrabilityTestSpin1995,gomborIntegrableSpinChains2021}.
Conversely, previous proofs of nonintegrability in such systems~\cite{shiraishiProofAbsenceLocal2019,chibaMixfield2024,chibaExactThermalEigenstates2024,yamaguchiCompleteClassificationIntegrability2024,yamaguchiProofAbsenceLocal2024,parkProofNonintegrabilitySpin12024,hokkyoProofAbsenceLocal2025} 
have supported the contrapositive~\footnote{
Strictly speaking, it is not literally \textit{contrapositive} 
because there can exist systems that are neither integrable nor nonintegrable, but partially integrable, i.e., with some nontrivial local conserved quantities. 
Surprisingly, the existence of such partially integrable models is denied in the models examined in previous studies of nonintegrability.} 
conjecture: the absence of $3$-local conserved quantities is a sufficient condition for nonintegrability~\cite{yamaguchiCompleteClassificationIntegrability2024}.
However, these tests for integrability and nonintegrability remain empirical, and, 
to the best of our knowledge, no fully rigorous proof has yet been established. 
The main goal of this paper is to present a rigorously provable test for integrability and nonintegrability based solely on $3$-local quantities. 
Our approach unifies the conventional, handcrafted methods used in the study of nonintegrability and provides a partial resolution to the above conjectures.
Using our rigorous test, 
we construct an algorithm--independent of system size--that guarantees the nonintegrability of translationally invariant quantum spin chains.

This paper is organized as follows. 
In Sec.~\ref{sec:setup}, we fix the notation and state the main theorem. 
As an application of the theorem, 
we demonstrate the absence of local conserved quantities in several systems in Sec.~\ref{sec:examples}. 
Section~\ref{sec:proof} is devoted to proving the theorems, 
where we also introduce a graphical language used in the discussion. 
Technical details of the proof are provided in Appendices. 
We also argue that our theorem and its proof can be applied to the analysis of a related structure: spectrum generating algebras. 
In Sec.~\ref{sec:algorithm}, we construct an algorithm to test integrability and nonintegrability in translationally invariant chains.
Sec.~\ref{sec:internal} presents an alternative test for systems where the method in Sec.~\ref{sec:setup} fails to apply.
We conclude the paper with a summary and outlook in Sec.~\ref{sec:summary}.

\section{Setup and main result}\label{sec:setup}
\subsection{Models}
We consider a quantum spin chain on $N$ sites subject to the periodic boundary condition. 
The set of sites is denoted by $\Lambda_N=\{1,\dots, N\}$, 
where we identify $N+i$ with $i$. 
Each site has a local state space $\Hilb\cong\mathbb{C}^d$, 
which is a $d$-dimensional complex Hilbert space with $d<\infty$. 
Accordingly, the total Hilbert space is $\Hilb_\tot\coloneqq\Hilb^{\otimes \Lambda_N}$. 
Note that the following argument apply to arbitrary finite
$d$, 
and are not restricted to spin-1/2 systems.

Next, we introduce $k$-local operators. 
Let $\bound$ denote the set of (bounded) linear operators on $\Hilb$, 
and $\bound_0$ the set of traceless operators on $\Hilb$. 
For each operator $\hat{X}$ in $\bound_\tot\coloneqq\bound^{\otimes\Lambda_N}$, 
its support is defined as the smallest set of sites on which $\hat{X}$ acts nontrivially. 
The set of (traceless) operators 
supported on a subset $A(\subset\Lambda_N)$ is defined as 
$\bound_0^{A}\coloneqq\bound_0^{\otimes A}$. 
Here we identify $\bound_0^{\otimes A}$ with $\bound_0^{\otimes A}\otimes I^{\otimes(\Lambda_N\setminus A)}$ where $I$ is the identity operator. 
The set of (strictly) $k$-local operators $\bound^{(k)}$ 
is defined as a (direct) sum of $\bound_0^{A}$ with $\size(A)=k$, 
where 
\begin{equation}
    \size(A)\coloneqq\min\{d\in\mathbb{Z}_{\geq 0}\mid {}^\exists l\in\mathbb{Z},\ A+l\subset\{1,\dots,d\}\}
\end{equation}
is the size of a subset $A$. 
We set $\bound^{(0)}=\bound_0^{ \emptyset}=\mathbb{C}I^{\otimes\Lambda_N}$. 
We define the set of at most $k$-local quantities $\tbound^{(k)}$ as a (direct) sum of $\bound^{(l)}$ for $0\leq l\leq k$. 
Since $\bound=\bound_0\oplus\mathbb{C}I$ holds, we have $\tbound^{(N)}=\bound_\tot$.
We define the length of an operator $\hat{X}$ as 
$\len(\hat{X})\coloneqq\min\{k\in\{0, \dots, N\}\mid \hat{X}\in\tbound^{(k)}\}$, 
which quantifies the degree of locality of $\hat{X}$. 
When we refer to the operator $\hat{X}$ simply as a $k$-local quantity, 
we mean that $\len(\hat{X})=k$. 
Such an operator can be decomposed into $l$-local terms for $0\leq l\leq k$: 
$\hat{X}=\sum_{l=0}^k \hat{X}^{(l)}$ with a nonvanishing $\hat{X}^{(k)}$, 
where $\hat{X}^{(l)}$ is an element of $\bound^{(l)}$. 
If we emphasize that $\hat{X}^{(l)}$ is a component of $k$-local quantity $\hat{X}$, 
we write $\hat{X}_{[k]}, \hat{X}_{[k]}^{(l)}$ instead of $\hat{X}, \hat{X}^{(l)}$. 
For $l\geq 1$, each $\hat{X}^{(l)}$ is further expanded as 
$\hat{X}^{(l)}=\sum_{i=1}^N\hat{X}^{(l)}_i$, 
where
\begin{equation}
    \hat{X}^{(l)}_i\in  \bigoplus_{\substack{\size(A)=l\\i,i+l-1\in A}}\bound_0^{A}
    =\bound_0^{\{i\}}\otimes\bound^{\otimes(l-2)}\otimes\bound_0^{\{i+l-1\}}
\end{equation}
for $l\geq 2$ and $\hat{X}^{(1)}_i\in\bound_0^{\{i\}}$. 

Finally, we introduce the Hamiltonian of the system, 
which consists of nearest-neighbor interactions and on-site potentials.
This corresponds to a $2$-local Hamiltonian, 
which is represented as 
\begin{equation}
    \hat{H}=\sum_{i=1}^N \hat{H}^{(2)}_{i}+\sum_{i=1}^N\hat{H}^{(1)}_i   %
    ,\label{eq:Hamiltonian}
\end{equation}
that is, $ \hat{H}^{(2)}_{i}\in\bound_0^{\{i,i+1\}}$ and $\hat{H}^{(1)}_i\in\bound_0^{\{i\}}$. 
We are mainly interested in translationally invariant systems, 
but we do not assume any relationship between these operators at this point. 
In addition, $\hat{H}$ can be non-hermitian. %

\subsection{Main Result}\label{subsec:main}
We impose the following assumption on $\hat{H}^{(2)}_{i}$:
\begin{equation}
 \begin{aligned} 
  \hat{X}\in\bound_0^{\{i+1\}},\ [I\otimes \hat{X},\hat{H}^{(2)}_{i}]=0
  &\Rightarrow\hat{X}=0\\
  \mbox{and}\ \hat{Y}\in\bound_0^{\{i\}},\ [\hat{Y}\otimes I,\hat{H}^{(2)}_{i}]=0
  &\Rightarrow\hat{Y}=0
  \end{aligned}
\label{eq:assumption_injective}
\end{equation}
for all $i\in\Lambda_N$. 
We refer to this condition as injectivity, 
which ensures that this Hamiltonian is not ``singular''. 
For example, let us consider an XYZ-type interaction: 
$\hat{H}^{(2)}_{i}=\sum_{\alpha=x,y,z}J_{\alpha}\hat{S}_{i}^{\alpha}\hat{S}_{i+1}^{\alpha}$ 
where $\hat{S}_{i}^{\alpha}$ represents the spin operator of spin $\S\in\mathbb{Z}_{>0}/2$. 
Then, for any $\S$, the injectivity~\eqref{eq:assumption_injective} is equivalent to requiring that
at least two elements in $\{J_x,J_y,J_z\}$ are nonzero. 
Otherwise, this interaction becomes an Ising-type one or vanishes. 

We are now in a position to state the main result of this paper: a rigorous test for integrability and nonintegrability. 
We first state a result 
that holds under an additional assumption that is satisfied in many systems of interest.
We introduce the following set of operators:
\begin{align}
    \bound_{\leq}^{(2)}&\coloneqq\{\hat{X}\in\bound^{(2)}\mid \len([\hat{X},\hat{H}])\leq 2\}\nonumber\\
    &=\{\hat{X}\in\bound^{(2)}\mid \len([\hat{X},\hat{H}^{(2)}])\leq 2\}.
\end{align}
It is evident that $\bound_{\leq}^{(2)}$ contains $\hat{H}^{(2)}$. 
Conversely, 
for generic spin chains with sufficiently complex interactions, 
it is expected that $\bound_{\leq}^{(2)}$ is limited to this, 
i.e., 
\begin{equation}
    \bound_{\leq}^{(2)}=\mathbb{C}\hat{H}^{(2)}\label{eq:assumption_1dimensionality}
\end{equation}
holds. 
For example, the spin-$S$ XYZ chains satisfy Eq.~\eqref{eq:assumption_1dimensionality}
for arbitrary $S\in\mathbb{Z}_{>0}/2$; 
see the last paragraph of Sec.~\ref{sec:algorithm} and Appx.~\ref{sec:general_XYZ}. 
Under these assumptions, 
we can show the following theorem rigorously. 
\begin{thm}[Main Result]\label{thm:main}
Consider a Hamiltonian $\hat{H}$ with nearest-neighbor interactions and on-site potentials 
(see Eq.~\eqref{eq:Hamiltonian})
satisfying the assumptions~\eqref{eq:assumption_injective} 
and~\eqref{eq:assumption_1dimensionality}. 
If there is no $3$-local quantity $\hat{Q}$ such that $\len([\hat{Q},\hat{H}])\leq 2$, 
then the system has no $k$-local conserved quantity for any $3\leq k\leq N/2$.
Equivalently, if there is a $k$-local conserved quantity for some $3\leq k\leq N/2$, 
there also exists a $3$-local quantity $\hat{Q}$ such that $\len([\hat{Q},\hat{H}])\leq 2$. 
\end{thm}
This theorem provides a sufficient condition for nonintegrability, 
or more precisely, 
the absence of $k$-local conserved quantities for $3\leq k\leq N/2$. 
Note that all models satisfying condition~\eqref{eq:assumption_injective}, 
among those whose nonintegrability has been proven~\cite{shiraishiProofAbsenceLocal2019,chibaExactThermalEigenstates2024,yamaguchiCompleteClassificationIntegrability2024,yamaguchiProofAbsenceLocal2024,parkProofNonintegrabilitySpin12024,hokkyoProofAbsenceLocal2025}, 
satisfy this sufficient condition.
Equivalently, 
Thm.~\ref{thm:main} implies that 
the existence of a 3-local quantity $\hat{Q}$ with $\len([\hat{Q},\hat{H}])\leq 2$ is a necessary condition for integrability. 
This yields a partial resolution of the conjecture in Ref.~\cite{grabowskiIntegrabilityTestSpin1995}, 
which states that the existence of a 3-local conserved quantity $\hat{Q}$, 
that is, a 3-local quantity satisfying $\len([\hat{Q},\hat{H}])=0$ 
is necessary for integrability. 
In summary, 
Thm.~\ref{thm:main} gives a unified perspective and greatly simplifies existing proofs of nonintegrability, 
and provides a rigorous foundation for empirical tests of integrability. 

Thm.~\ref{thm:main} can be applied to models with finite-range interactions, 
because such a model is mathematically equivalent to some spin chain with a larger local dimension. 
Note that there is an extension of the above conjecture to medium-range interactions 
rather than relying on such a transformation to nearest-neighbor interactions~\cite{gomborIntegrableSpinChains2021}.
We can also regard spin ladders and systems on higher-dimensional lattices as one-dimensional systems in the same way~\footnote{
In this case, the ``locality'' of an operator is defined in a very weak sense; it is determined by the size of its support for some direction. 
}; see, e.g., Subsec.~\ref{subsec:triangular}. 
While Eq.~\eqref{eq:assumption_1dimensionality} is expected to be satisfied if the interaction is sufficiently complex, 
it does not hold for models with nearest-neighbor interactions on the hypercubic lattice or the honeycomb lattice. 
Instead, we will explain a method to simplify the proof of nonintegrability in such systems in Sec.~\ref{sec:internal}.

Next, we consider general systems with the injectivity, 
which may not satisfy the assumption~\eqref{eq:assumption_1dimensionality}. 
We can prove some results even in this case, 
which limit a possible form of conserved quantities.
As generalizations of $\bound_{\leq}^{(2)}$, 
we introduce the following sets of operators:
\begin{align}
    &\bound_{\leq}^{(k)}\coloneqq\{\hat{X}\in\bound^{(k)}\mid \len([\hat{X},\hat{H}])\leq k\},\\
    &\bound_{<}^{(k)}
    \coloneqq
    \{\hat{X}\in\bound^{(k)}_\leq\oplus\bound^{(k-1)}\mid \len([\hat{X},\hat{H}])\leq k-1\}.
\end{align}
If it is necessary to indicate the dependence on $\hat{H}$, we denote $\bound_{\leq}^{(k)}$ and $\bound_{<}^{(k)}$ by $\bound_{\leq}^{(k)}(\hat{H})$ and $\bound_{<}^{(k)}(\hat{H})$, respectively.
Note that, since $\len([\hat{X},\hat{H}])\leq\len(\hat{X})+1$ holds, 
$\hat{X}_{[k]}^{(k)}\in\bound_{\leq}^{(k)}$ 
and $\hat{X}_{[k]}^{(k)}+\hat{X}_{[k]}^{(k-1)}\in\bound_{<}^{(k)}$ %
are the easiest and second easiest conditions 
that must be satisfied for a $k$-local quantity 
$\hat{X}_{[k]}=\sum_{l=0}^k \hat{X}_{[k]}^{(l)}$ to be a conserved quantity.
For $\bound_{\leq}^{(k)}$'s, we have the following theorem. 
\begin{thm}\label{thm:general_step1}
Consider a Hamiltonian $\hat{H}$ with nearest-neighbor interactions and on-site potentials 
(see Eq.~\eqref{eq:Hamiltonian})
which satisfies the assumption~\eqref{eq:assumption_injective}. 
Then, we have the following linear isomorphisms between $\bound_{\leq}^{(k)}$'s: 
    \begin{equation}
        \iota_k:\bound_{\leq}^{(k-1)}\xrightarrow{\cong}\bound_{\leq}^{(k)}
        \label{eq:Step1}
    \end{equation}
for $3\leq k\leq N/2$, 
which acts on $\hat{X}^{(k)}=\sum_{i=1}^N\hat{X}^{(k-1)}_i$ in $\bound_{\leq}^{(k-1)}$ as
\begin{equation}
    \hat{X}^{(k-1)}
    \mapsto\sum_{i=1}^N[\hat{X}^{(k-1)}_i,\hat{H}^{(2)}_{i+k-2}]\in\bound_{\leq}^{(k)}.\label{eq:boost}
\end{equation}
\end{thm}
Theorem~\ref{thm:general_step1} provides an insight into 
how to use Thm.~\ref{thm:main} for proving nonintegrability; 
first, we determine $\bound_{\leq}^{(2)}$
and check the assumption~\eqref{eq:assumption_1dimensionality}. 
Next, we search for the solution of $\len([\iota_2(\hat{H}^{(2)})+\hat{X}^{(2)},\hat{H}])\leq 2$ 
for $\hat{X}^{(2)}\in\bound^{(2)}$. 
If there is no such solution, we can conclude that the system is nonintegrable. 
This protocol not only unifies but also greatly simplifies 
the existing proofs of nonintegrability of systems satisfying~\eqref{eq:assumption_injective}
and~\eqref{eq:assumption_1dimensionality}~\cite{shiraishiProofAbsenceLocal2019,yamaguchiCompleteClassificationIntegrability2024,yamaguchiProofAbsenceLocal2024,parkProofNonintegrabilitySpin12024,hokkyoProofAbsenceLocal2025}, 
where the absence of $k$-local conserved quantities is discussed separately for each $k$. 
A more specific procedure for translationally invariant systems, 
including how to determine $\bound_{\leq}^{(2)}$, will be discussed in Sec.~\ref{sec:algorithm}.

We note that the action of Eq.~\eqref{eq:boost} is similar to that of the boost operator, 
which produces a tower of conserved quantities in Bethe solvable models~\cite{thackerCornerTransferMatrices1986,linksLadderOperatorOnedimensional2001}. 
Even when assumption~\eqref{eq:assumption_1dimensionality} fails,  
one can still prove the absence of conserved quantities generated 
by such boost from the Hamiltonian.
By using the above isomorphisms, 
$\hat{H}^{(2)}\in\bound_{\leq}^{(2)}$ is mapped to some operator 
\begin{equation}
    \hat{Q}_b^{(k)}
    \coloneqq
    \iota_{k-1}\circ\dots\circ \iota_2(\hat{H}^{(2)})
    \in\bound_{\leq}^{(k)}
    \label{eq:boosted_operator}
\end{equation}
for $3\leq k\leq N/2$. 
These operators define a subspace of $\bound_{<}^{(k)}$ as 
\begin{equation}
    \bound_{<,b}^{(k)}\coloneqq\bound_{<}^{(k)}\cap(\mathbb{C}\hat{Q}_b^{(k)}\oplus\bound^{(k-1)}), 
\end{equation}
which represents the set of operator in $\bound_{<}^{(k)}$ whose $k$-local component is proportional to $\hat{Q}_b^{(k)}$. 
Then, we have the following theorem. 
\begin{thm}\label{thm:general_step2}
Let us consider a Hamiltonian $\hat{H}$ of the  same form as in Thm.~\ref{thm:general_step1}. 
Then, we have the following linear isomorphisms between $\bound_{<,b}^{(k)}$'s: 
    \begin{equation}
        \bound_{<,b}^{(3)}
        \xrightarrow{\cong}\bound_{<,b}^{(4)}\xrightarrow{\cong}
        \dots\xrightarrow{\cong}
        \bound_{<,b}^{(\lfloor N/2\rfloor)},\label{eq:Step2}
    \end{equation}
which map the subspace $\bound_{\leq}^{(k-1)}\subset\bound_{<,b}^{(k)}$ 
into $\bound_{\leq}^{(k)}\subset\bound_{<,b}^{(k+1)}$. 
In particular, if $\bound_{<,b}^{(3)}=\bound_{\leq}^{(2)}$ holds, 
then there is no $k$-local conserved quantity in $\bound_{<,b}^{(k)}$ for $3\leq k\leq N/2$. 
\end{thm}
If the assumption~\eqref{eq:assumption_1dimensionality} holds, 
we have $\bound_{<,b}^{(k)}=\bound_{<}^{(k)}$ and 
Thm.~\ref{thm:general_step2} yields the stronger result than Thm.~\ref{thm:main}: 
if there is no $3$-local quantity $\hat{Q}_{[3]}$ such that $\len([\hat{Q}_{[3]},\hat{H}])\leq 2$, 
then there is no $k$-local quantity $\hat{Q}_{[k]}$ 
such that $\len([\hat{Q}_{[k]},\hat{H}])\leq k-1$ for any $3\leq k\leq N/2$.

\section{Examples}\label{sec:examples}
As applications of Thm.~\ref{thm:main} 
and Thm.~\ref{thm:general_step1}, 
we prove the nonintegrability of several quantum spin systems.

\subsection{The XXZ chain with a non-uniform magnetic field}\label{subsec:XXZ}
We show the absence of local conserved quantities in the spin-1/2 XXZ chain with a non-uniform magnetic field. 
The Hamiltonian of this system is represented as
\begin{equation}
    \hat{H}=\sum_{i=1}^N 
     \sum_{\alpha\in\{x,y,z\}}
     J_{\alpha}\hat{\sigma}_i^\alpha\hat{\sigma}_{i+1}^\alpha
    +\sum_{i=1}^N h_i\hat{\sigma}_i^z,\label{eq:non-uniform_XXZ}
\end{equation}
where $\hat{\sigma}^{\alpha}$'s represent the Pauli matrices. 
We assume that $J_x=J_y\neq0$ and $J_z\neq0$. 
This model fits the form of Eq.~\eqref{eq:Hamiltonian}, 
where $\hat{H}^{(2)}_i$ represents the XXZ interaction and $\hat{H}^{(1)}_i$ 
corresponds to the magnetic field. 
We refer to the magnetic field as \textit{uniform} if $h_i$ is independent of $i$
; otherwise, it is \textit{non-uniform}. 
If a magnetic field is uniform, this system is a Bethe solvable system~\cite{yangOneDimensionalChainAnisotropic1966}, 
and local conserved quantities are explicitly constructed in Refs.~\cite{grabowskiStructureConservationLaws1995,nozawaExplicitConstructionLocal2020}. 

To analyze the existence of local conserved quantities, 
we first determine $\bound_{\leq}^{(2)}$. 
This step has been carried out in Refs.~\cite{shiraishiProofAbsenceLocal2019,yamaguchiCompleteClassificationIntegrability2024,yamaguchiProofAbsenceLocal2024}, 
and it was shown that $\bound_{\leq}^{(2)}=\mathbb{C}\hat{H}^{(2)}$, 
which allows us to apply Thm.~\ref{thm:main}. 
By using Eq.~\eqref{eq:Step1}, 
we can determine $\bound_{\leq}^{(k)}$; 
all of them are one-dimensional. 
In particular, $\bound_{\leq}^{(3)}$ is generated by the following operator:
\begin{align}
    \hat{Q}_{[3]}^{(3)}&\coloneqq\frac{\im}{2}\iota_2(\hat{H}^{(2)})\nonumber\\
    &=\sum_{i=1}^N 
    \sum_{\alpha,\beta,\gamma\in\{x,y,z\}}
    J_{\alpha}J_{\gamma}\varepsilon_{\alpha\beta\gamma}
    \hat{\sigma}_i^\alpha\hat{\sigma}_{i+1}^\beta\hat{\sigma}_{i+2}^\gamma.
    \label{eq:Q3}
\end{align} 
Here, $\varepsilon_{\alpha\beta\gamma}$ denotes the Levi-Civita symbol. 

Next, we determine $\bound_{<}^{(3)}$. 
In the following, we show that this set is trivial, 
that is, $\bound_{<}^{(3)}\subset\bound^{(2)}$ 
if a magnetic field is non-uniform. 
To prove this, we derive a contradiction 
by assuming that 
\begin{equation}
    \len([\hat{Q}_{[3]}^{(3)}+\hat{Q}_{[3]}^{(2)},\hat{H}])\leq 2\label{eq:XXZ_assumption}
\end{equation}
holds 
for some $\hat{Q}_{[3]}^{(2)}\in\bound^{(2)}$. 
This operator can be expanded in terms of the Pauli operators as
\begin{equation}
    \hat{Q}_{[3]}^{(2)}=\sum_{i=1}^N\sum_{\alpha,\beta\in\{x,y,z\}}
    q_{\alpha\beta,i}\hat{\sigma}_i^\alpha\hat{\sigma}_{i+1}^\beta. 
\end{equation}
Now we expand $[\hat{Q}_{[3]}^{(3)}+\hat{Q}_{[3]}^{(2)},\hat{H}]$ as 
\begin{align}
    [\hat{Q}_{[3]}^{(3)}+\hat{Q}_{[3]}^{(2)},\hat{H}]
    &=
    2\im\sum_{i=1}^N\sum_{\alpha,\beta,\gamma\in\{x,y,z\}}
    c_{\alpha\beta\gamma,i}\hat{\sigma}_i^{\alpha}\hat{\sigma}_{i+1}^{\beta}\hat{\sigma}_{i+2}^{\gamma}\nonumber\\
    &+\mbox{(a 2-local quantity)}, 
\end{align}
where $c_{\alpha\beta\gamma,i}$ is a function of $J_{\alpha'}$ and $q_{\alpha'\beta',i}$. 
The assumption~\eqref{eq:XXZ_assumption} is equivalent to $c_{\alpha\beta\gamma,i}=0$ for all $\alpha,\beta,\gamma$ and $i\in\Lambda_N$. 
First, we calculate $c_{xxz,i}$. 
This coefficient comes from the following three commutators:
\begin{align}
    &[J_zJ_x\hat{\sigma}_i^{x}\hat{\sigma}_{i+1}^{y}\hat{\sigma}_{i+2}^{z}, h_{i+1}\hat{\sigma}_{i+1}^{z}],
    \\
    &[-J_zJ_x\hat{\sigma}_i^{y}\hat{\sigma}_{i+1}^{x}\hat{\sigma}_{i+2}^{z}, h_{i}\hat{\sigma}_{i}^{z}],\ \mbox{and}
    \\
    &[q_{xy,i}\hat{\sigma}_i^{x}\hat{\sigma}_{i+1}^{y},J_z\hat{\sigma}_{i+1}^{z}\hat{\sigma}_{i+2}^{z}], 
\end{align}
which determine the form of $c_{xxz,i}$: 
\begin{equation}
    c_{xxz,i}=J_z[J_x(h_{i+1}-h_i)+q_{xy,i}]. 
\end{equation}
Next, we calculate $c_{zyy,i-1}$. 
This coefficient comes from the following three commutators:
\begin{align}
    &[J_zJ_x\hat{\sigma}_{i-1}^{z}\hat{\sigma}_{i}^{x}\hat{\sigma}_{i+1}^{y}, h_{i}\hat{\sigma}_{i}^{z}],
    \\
    &[-J_zJ_x\hat{\sigma}_{i-1}^{z}\hat{\sigma}_{i}^{y}\hat{\sigma}_{i+1}^{x}, h_{i+1}\hat{\sigma}_{i+1}^{z}],\ \mbox{and}
    \\
    &[q_{xy,i}\hat{\sigma}_i^{x}\hat{\sigma}_{i+1}^{y},J_z\hat{\sigma}_{i-1}^{z}\hat{\sigma}_{i}^{z}], 
\end{align}
which determine the form of $c_{zyy,i-1}$: 
\begin{equation}
    c_{zyy,i-1}=J_z[J_x(h_{i+1}-h_i)-q_{xy,i}]. 
\end{equation}
Therefore, $c_{xxz,i}= c_{zyy,i-1}=0$ implies $h_{i+1}=h_i$, where we use $J_x,J_z\neq0$. 
This contradicts with non-uniformity of a magnetic field, and we have $\bound_{<}^{(3)}\subset\bound^{(2)}$. 
By applying Thm.~\ref{thm:main}, 
it turns out that there is no $k$-local conserved quantity for $3\leq k\leq N/2$.

In summary,we obtain the following result.
\begin{prop}\label{prop:XXZ}
    Any model in Eq.~\eqref{eq:non-uniform_XXZ} with a non-uniform magnetic field has no $k$-local conserved quantity for $3\leq k\leq N/2$. 
\end{prop}
Although nonintegrability of systems without spatial uniformity is also discussed in Ref.~\cite{shiraishiAbsenceLocalConserved2024}, 
the above result provides the first example 
in which it is rigorously shown that 
the system is nonintegrable \textit{only if} it is not uniform.

We comment on the connection between the above result and thermalization of this system.
Many-body localization (MBL) is a phenomenon where strong disorder prevents a system from thermalizing, leading to a breakdown of the eigenstate thermalization hypothesis (ETH)
~\cite{nandkishoreManyBodyLocalizationThermalization2015,altmanUniversalDynamicsRenormalization2015,abaninRecentProgressManybody2017,aletManybodyLocalizationIntroduction2018}. 
Whether the XXZ model with a random magnetic field exhibits MBL in the thermodynamic limit is a subject of debate~\cite{serbynLocalConservationLaws2013,luitzManybodyLocalizationEdge2015,znidaricInteractionInstabilityLocalization2018,suntajsQuantumChaosChallenges2020,kiefer-emmanouilidisEvidenceUnboundedGrowth2020,selsDynamicalObstructionLocalization2021,kiefer-emmanouilidisSlowDelocalizationParticles2021,selsThermalizationDiluteImpurities2023,imbrieManyBodyLocalizationQuantum2016,roeckAbsenceNormalHeat2024}.
The MBL is attributed to the existence of an extensive number of quasi-local conserved quantities~\cite{serbynLocalConservationLaws2013,husePhenomenologyFullyManybodylocalized2014}, 
which are operators consisting of $k$-local operators ($0\leq k\leq N$)
whose coefficients decay exponentially with $k$. 
Such a mechanism is sometimes referred to as the emergence of local integrability~\cite{abaninRecentProgressManybody2017,aletManybodyLocalizationIntroduction2018}, 
but note that this is different from the integrability we discuss, 
in the sense that there is an extensive number of (strictly) local conserved quantities. 
In fact, Prop.~\ref{prop:XXZ} tells us that the XXZ model with a random magnetic field is nonintegrable in our sense. 

\subsection{The $S=\frac{1}{2}$ XYZ model on the triangular lattice}\label{subsec:triangular}
To illustrate our main result in higher-dimensional systems, 
we consider the $S = \frac{1}{2}$ XYZ model on the triangular lattice. 
By embedding the triangular lattice into the square lattice (see Fig.~\ref{fig:triangular}), 
the Hamiltonian of the system takes the following form:
\begin{equation}
     \hat{H}=\sum_{\substack{\mathbf{r}\in\Lambda_N\times\Lambda_M
    \\
    \mathbf{v}=(1,0),(1,1),(0,1)}}
     \sum_{\alpha\in\{x,y,z\}}
     J_{\alpha}\hat{\sigma}_{\mathbf{r}}^\alpha
     \hat{\sigma}_{\mathbf{r}+\mathbf{v}}^\alpha
    ,\label{eq:triangle_XYZ}
\end{equation}
where we assume that $M\geq3$. 
\begin{figure}[tbp]
    \centering
    \includegraphics[width=1\linewidth]{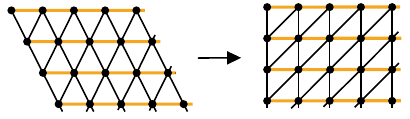}
    \caption{The triangular lattice can be regarded as 
    the square lattice with additional edges.  
    We define locality along the direction indicated by the horizontal lines.
    }\label{fig:triangular}
\end{figure}
We impose periodic boundary conditions (PBC) in one direction and 
either PBC or open boundary conditions (OBC) in the other: 
\begin{align}
    \hat{\sigma}_{(N+1,j)}^\alpha&=\hat{\sigma}_{(1,j)}^\alpha,\nonumber\\
    \hat{\sigma}_{(i,M+1)}^\alpha&=
    \begin{cases}
        \hat{\sigma}_{(i,1)}^\alpha\ (\mbox{PBC});\\
        0\ (\mbox{OBC}).
    \end{cases}
\end{align}

We define locality by treating $\{i\} \times \Lambda_M$ as a single site in an effectively one-dimensional system.
According to Eq.~\eqref{eq:Hamiltonian}, the Hamiltonian is divided as 
\begin{align}
    \hat{H}_i^{(2)}&=
    \sum_{\substack{j\in\Lambda_M\\\mathbf{v}=(1,0),(1,1)}}
     \sum_{\alpha\in\{x,y,z\}}
     J_{\alpha}\hat{\sigma}_{(i,j)}^\alpha
     \hat{\sigma}_{(i,j)+\mathbf{v}}^\alpha,\\
     \hat{H}_i^{(1)}&=
    \sum_{j\in\Lambda_M}
     \sum_{\alpha\in\{x,y,z\}}
     J_{\alpha}\hat{\sigma}_{(i,j)}^\alpha
     \hat{\sigma}_{(i,j+1)}^\alpha.
\end{align}
Based on the above definition, 
we have the following proposition. 
\begin{prop}\label{prop:triangle}
    The $S=\frac{1}{2}$ XYZ model on the triangular lattice~\eqref{eq:triangle_XYZ}
    has no $k$-local conserved quantity for $3\leq k\leq N/2$. 
\end{prop}
\begin{proof}
We prove Prop.~\ref{prop:triangle} by using Thm.~\ref{thm:main} and Thm.~\ref{thm:general_step1}. 
\subsubsection{Injectivity of $\hat{H}_i^{(2)}$}
Consider an operator $\hat{X}\in\bound_0^{\{i\}\times\Lambda_M}$ 
such that $[\hat{X}\otimes I,\hat{H}_i^{(2)}]=0$. 
Our goal in this step is to show that $\hat{X}=0$. 
We expand $\hat{X}$ and $[\hat{X}\otimes I,\hat{H}_i^{(2)}]$ by the Pauli basis: 
\begin{align}
    \hat{X}=
    \sum_{\substack{\bm{\alpha}\in\{0,x,y,z\}^{\Lambda_M}\\\bm{\alpha}\neq\bm{0}}}
    q_{\bm{\alpha}}\bigotimes_{j\in\Lambda_M}\hat{\sigma}_{(i,j)}^{\alpha_j}
    \eqqcolon 
    \sum_{\bm{\alpha}}q_{\bm{\alpha}}\hat{\bm{\sigma}}_i^{\bm{\alpha}}, \\
    [\hat{X}\otimes I,\hat{H}_i^{(2)}]=\sum_{\substack{\bm{\alpha}\\\beta\in\{x,y,z\}\\j'\in\Lambda_M}}
    r_{\bm{\alpha},\beta;j'}
    \hat{\bm{\sigma}}_i^{\bm{\alpha}}\hat{\sigma}_{(i+1,j')}^{\beta}, 
\end{align}
where $\hat{\sigma}_0\coloneqq I$ is the identity operator. 
When we consider the OBC, we take $\alpha_{0}=\alpha_{M+1}=0$. 

First, we show $q_{\bm{\alpha}}=0$ if $\alpha_j=0$ for some $j\in\Lambda_M$. 
Since $\bm{\alpha}\neq\bm{0}$, 
we can appropriately redefine $j$ so that either $\alpha_{j+1} \neq 0$ or $\alpha_{j-1} \neq 0$ holds.
If $\alpha_{j+1}=x$, then we have 
\begin{equation}
    r_{\bm{\alpha}',y;j+1}=2\im J_yq_{\bm{\alpha}},
\end{equation}
where
\begin{equation}
    \alpha'_{j'}=
    \begin{cases}
    \alpha_{j'} & (j'\neq j+1)\\
    z & (j'=j+1)
    \end{cases}.\label{eq:alpha'}
\end{equation}
Therefore, $[\hat{X}\otimes I,\hat{H}_i^{(2)}]$ implies $q_{\bm{\alpha}}=0$. 
The remaining cases follow similarly. 

Next, we consider $\bm{\alpha} \in \{x,y,z\}^{\Lambda_M}$ 
such that there exist $j, j'\in\Lambda_M$ satisfying $\alpha_j \neq \alpha_{j'}$.
In this case, we can take some $1\leq j<M$ satisfying $\alpha_j\neq\alpha_{j+1}$. 
If $(\alpha_j,\alpha_{j+1})=(y,x)$, we have $r_{\bm{\alpha}',y;j+1}=2iJ_yq_{\bm{\alpha}}$ 
for $\bm{\alpha}'$ in Eq.~\eqref{eq:alpha'}. 
The same applies to other cases, and we have shown $q_{\bm{\alpha}}=0$ 
if $\alpha_j \neq \alpha_{j'}$ for some $(j,j')$. 

Finally, we examine the case
where $\alpha_j = \alpha_1\in\{x,y,z\}$ for all $j\in\Lambda_M$.
If $\alpha_1=x$, we have
\begin{equation}
    r_{\bm{\alpha}',y;2}=2\im J_y(q_{\bm{\alpha}}-q_{\tilde{\bm{\alpha}}}),
\end{equation}
where $\bm{\alpha}'$ is defined in Eq.~\eqref{eq:alpha'} and 
\begin{equation}
    \tilde{\alpha}_{j}=
    \begin{cases}
    z & (j= 1,2)\\
    x & (j\neq 1,2)
    \end{cases}.
\end{equation}
Therefore, $[\hat{X}\otimes I,\hat{H}_i^{(2)}]$ implies $q_{\bm{\alpha}}=q_{\tilde{\bm{\alpha}}}$. 
On the other hand, we have shown $q_{\tilde{\bm{\alpha}}}=0$ in the above. 
Thus, $q_{\bm{\alpha}}=0$. 
The same argument can be made for $\alpha_1 = y,z$, 
and it follows that $q_{\bm{\alpha}} = 0$.
From the above, it follows that $\hat{X} = 0$, and the first line of Eq.~\eqref{eq:assumption_injective} is proven. 
The proof for the second line can be done in a similar manner.

\subsubsection{Confirmation of Assumption~\eqref{eq:assumption_1dimensionality}}
In this step, we confirm the assumption~\eqref{eq:assumption_1dimensionality}, 
or more specifically, 
we show that $\hat{X}\propto\hat{H}^{(2)}$ for $\hat{X}\in\bound_\leq^{(2)}$. 
First, we note that $\hat{X}\in\bound_\leq^{(2)}$ is equivalent to 
\begin{equation}
    [\hat{X}_i\otimes I,I\otimes\hat{H}^{(2)}_{i+1}]=[\hat{H}^{(2)}_i\otimes I,I\otimes\hat{X}_{i+1}]\label{eq:step1_operator}
\end{equation}
for all $i\in\Lambda_N$. 
From this equation and the injectivity, 
we know~\footnote{
See the discussion at the last paragraph of Sec.~\ref{sec:algorithm}. 
In this case, $\mathcal{A}$ can be taken as the subspace consisting of one-body operators. 
}
that $\hat{X}_i$ consists of two-body operators:
\begin{equation}
    \hat{X}_i=\sum_{\substack{\alpha,\beta\in\{x,y,z\}\\j,k\in\Lambda_M}}q_{\alpha\beta;jk,i}
    \hat{\sigma}^\alpha_{(i,j)}\hat{\sigma}^\beta_{(i+1,k)}. 
\end{equation}
Furthermore, it is clear that $q_{\alpha\beta;jk,i}=0$ unless $k \in\{j,j+1\}$.

Next, we show that $q_{\alpha\beta;jk,i}=q\delta_{\alpha\beta}J_\alpha$ for some $q\in\mathbb{C}$. 
By comparing the coefficient of $\hat{\sigma}^x_{(i,j)}\hat{\sigma}^x_{(i+1,k)}\hat{\sigma}^z_{(i+2,k)}$ on both sides of Eq.~\eqref{eq:step1_operator}, we have 
$2iJ_zq_{xy;jk,i}=0$. 
Similarly, we can show that $q_{\alpha\beta;jk,i}=0$ for $\alpha\neq\beta$. 
To consider the case of $\alpha=\beta$, 
we compare the coefficient of $\hat{\sigma}^x_{(i,j)}\hat{\sigma}^z_{(i+1,k)}\hat{\sigma}^y_{(i+2,l)}$ on both sides of Eq.~\eqref{eq:step1_operator}. 
This yields relations between $q_{\alpha\alpha;jk,i}$'s as
\begin{equation}
    J_yq_{xx;jk,i}=J_xq_{yy;kl,i+1}
\end{equation}
for all $i\in\Lambda_N$, $k\in\{j,j+1\}\cap\Lambda_M$ and $l\in\{k,k+1\}\cap\Lambda_M$. 
Similarly, we have 
\begin{equation}
    q_{\alpha\alpha;jk,i}/J_\alpha=q_{\beta\beta;kl,i+1}/J_\beta,
\end{equation}
which is equivalent to $q_{\alpha\alpha;jk,i}=\frac{q_{xx;11,1}}{J_x}J_\alpha$ 
for all $i\in\Lambda_N$ and $k\in\{j,j+1\}\cap\Lambda_M$. 
Hence, we have shown that $\hat{X}=\frac{q_{xx;11,1}}{J_x}\hat{H}^{(2)}$. 
Therefore, the Hamiltonian~\eqref{eq:triangle_XYZ} satisfies the assumption~\eqref{eq:assumption_1dimensionality}. 

\subsubsection{Absence of $3$-local quantity $\hat{Q}_{[3]}$ 
satisfying $\len([\hat{Q}_{[3]},\hat{H}])\leq 2$}
Up to this point, we have confirmed that 
the Hamiltonian~\eqref{eq:triangle_XYZ} satisfies the assumption~\eqref{eq:assumption_injective} and~\eqref{eq:assumption_1dimensionality}.  
Thus, Thm.~\ref{thm:main} can be applied, 
and it suffices to show that there is no $3$-local quantity $\hat{Q}$ 
such that $\len([\hat{Q},\hat{H}])\leq 2$ holds. 
By the fact~\eqref{eq:assumption_1dimensionality} and Thm.~\ref{thm:general_step1}, 
a $3$-local quantity $\hat{Q}_{[3]}$ 
satisfying $\len([\hat{Q}_{[3]},\hat{H}])\leq 2$ can be written as 
\begin{equation}
    \hat{Q}_{[3]}=q\sum_{i=1}^N[\hat{H}_i^{(2)},\hat{H}_{i+1}^{(2)}]+\hat{Q}_{[3]}^{(2)}
\end{equation}
with some $q\in\mathbb{C}$ and 2-local quantity $\hat{Q}_{[3]}^{(2)}$. 
Our goal is to show that $q=0$. 

We expand this operator $\hat{Q}_{[3]}^{(2)}$ by the Pauli basis:
\begin{equation}
    \hat{Q}_{[3],i}^{(2)}=\sum_{\bm{\alpha},\bm{\beta}\in\{0,x,y,z\}^{\Lambda_M}}
    q^{(2)}_{\bm{\alpha}\bm{\beta},i}
    \hat{\bm{\sigma}}^{\bm{\alpha}}_i\hat{\bm{\sigma}}^{\bm{\beta}}_{i+1}. 
\end{equation}
We focus on the following $(\bm{\alpha}, \bm{\beta})$:
\begin{align}
    \alpha_j&=
    \begin{cases}
        z & (j=1) \\
        x & (j=2) \\
        0 & (j\neq1,2)
    \end{cases},\\
    \beta_j&=
    \begin{cases}
        y & (j=1) \\
        0 & (j\neq1)
    \end{cases}.
\end{align}
The coefficient of $\hat{\sigma}_{(1,1)}^x\hat{\sigma}_{(2,1)}^y\hat{\sigma}_{(2,2)}^x\hat{\sigma}_{(3,1)}^y$ 
in $[\hat{Q}_{[3]},\hat{H}]$ comes from the following commutators:
\begin{align}
    [q^{(2)}_{\bm{\alpha}\bm{\beta},2}
    \hat{\sigma}_{(2,1)}^z\hat{\sigma}_{(2,2)}^x\hat{\sigma}_{(3,1)}^y,
    J_x\hat{\sigma}_{(1,1)}^x\hat{\sigma}_{(2,1)}^x], \\
    [2\im q J_xJ_y
    \hat{\sigma}_{(1,1)}^x\hat{\sigma}_{(2,1)}^z\hat{\sigma}_{(3,1)}^y,
    J_x\hat{\sigma}_{(2,1)}^x\hat{\sigma}_{(2,2)}^x]. 
\end{align}
Since the coefficient of $\hat{\sigma}_{(1,1)}^x\hat{\sigma}_{(2,1)}^y\hat{\sigma}_{(2,2)}^x\hat{\sigma}_{(3,1)}^y$ is zero 
by the assumption $\len([\hat{Q}_{[3]},\hat{H}])\leq 2$, 
the following equation holds:
\begin{equation}
    q^{(2)}_{\bm{\alpha}\bm{\beta},2}+2\im qJ_xJ_y=0.
\end{equation}
Similarly, 
by focusing on the coefficient of 
$\sigma_{(2,1)}^z\sigma_{(2,2)}^x\sigma_{(3,1)}^z\sigma_{(4,1)}^x$, 
we have 
\begin{equation}
    q^{(2)}_{\bm{\alpha}\bm{\beta},2}-2\im qJ_xJ_y=0.
\end{equation}
By combining these equations, we have $q=0$. 
Therefore, there is no $3$-local quantity $\hat{Q}_{[3]}$ 
satisfying $\len([\hat{Q}_{[3]},\hat{H}])\leq 2$ 
and this system has no $k$-local conserved quantity for $3\leq k\leq N/2$.  
\end{proof}

\section{Outline of Proof}\label{sec:proof}
\subsection{String Diagram}
Before proving the theorems, 
we introduce a ``basis-independent'' approach for a proof of nonintegrability. 
We adopt a string diagram as a useful graphical language, 
which is widely used in various areas of physics and mathematics~\cite{penroseApplicationsNegativeDimensional1971, baezPhysicsTopologyLogic2011,coeckePicturingQuantumProcesses2017}. 
The basic idea is to represent a linear map using boxes and wires. 
For instance, $\hat{X}\in\bound\otimes\bound_0$, $f:\bound_0^{\otimes 3}\to\bound_0^{\otimes 2}$, 
and $(\id_{\bound}\otimes f)(\hat{X}\otimes \bullet):\bound_0^{\otimes 2}\to\bound\otimes\bound_0^{\otimes 2}$ are depicted as
\begin{equation}
\figurecentering{\includegraphics[]{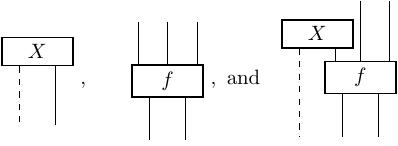}}
\ ,
\end{equation}
respectively, where $\id_{\bound}$ is the identity map on $\bound$. 
Here, each box represents a linear map: 
the top line corresponds to the input, and the bottom line to the output. 
In other words, time flows from top to bottom. 
Note that an element of a linear space is identified with a linear map from $\mathbb{C}$, 
which is represented by the ``empty'' line. 
The horizontal direction corresponds to the spatial direction, 
and placing the boxes horizontally corresponds to their tensor product.
In the following, 
solid and dashed lines represent $\bound_0$ 
and $\bound$, respectively.
Using this diagrammatic notation, the injectivity of 
$[\hat{H}^{(2)}_{i},\bullet\otimes I]$ and $[\hat{H}^{(2)}_{i},I\otimes \bullet]$ 
(see Eq.~\eqref{eq:assumption_injective})
are equivalent to the existence of linear maps $\mu_i,\lambda_i:\bound_0^{\otimes 2}\to\bound_0$ satisfying
\begin{equation}
\figurecentering{\includegraphics[]{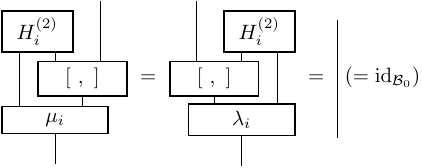}}
,
\label{eq:injective}
\end{equation}
for all $i\in\Lambda_N$. 
Here, $[\ ,\ ]:\bound_0^{\otimes 2}\to\bound_0$ denotes the commutator on $\bound_0$. 

\subsection{Isomorphisms between $\bound^{(k)}_{\leq}$'s}
Now we move on to the proof of Thm.~\ref{thm:general_step1}. 
It is sufficient to show the following proposition. 
\begin{prop}\label{prop:step1}
We fix an integer $k$ such that $3\leq k\leq N/2$. 
For each $\hat{Q}_{[k]}^{(k)}=\sum_{i=1}^N\hat{Q}^{(k)}_{[k],i}\in\bound_{\leq}^{(k)}$, 
there (uniquely) exists 
$\hat{Q}_{[k-1]}^{(k-1)}=\sum_{i=1}^N\hat{Q}^{(k-1)}_{[k-1],i}\in\bound_{\leq}^{(k-1)}$ 
satisfying
\begin{equation}
\figurecentering{\includegraphics[]{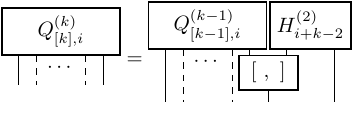}}
.
\label{eq:Step1_reduction}
\end{equation}
\end{prop}
\begin{proof}
The assumption $\len([\hat{Q}_{[k]}^{(k)},\hat{H}])\leq k$ is equivalent to the following equations: 
\begin{equation}
\figurecentering{\includegraphics[]{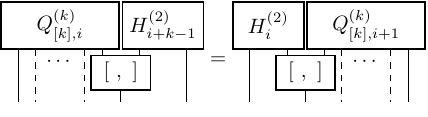}}
\label{eq:Step1_k}
\end{equation}
for all $i\in\Lambda_N$. 
This equation shows that the second wire from the left in $\hat{Q}^{(k)}_{[k],i}$ is, 
in fact, a solid line. 
Furthermore, by using the injectivity~\eqref{eq:injective}, 
we find that Eq.~\eqref{eq:Step1_reduction} is satisfied by taking $\hat{Q}^{(k-1)}_{[k-1],i}$ as
\begin{equation}
\figurecentering{\includegraphics[]{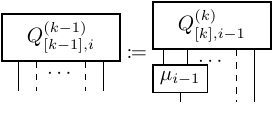}}
.
\end{equation}
By substituting Eq.~\eqref{eq:Step1_reduction} into Eq.~\eqref{eq:Step1_k} and using the injectivity~\eqref{eq:injective}, 
we obtain 
\begin{equation}
\figurecentering{\includegraphics[]{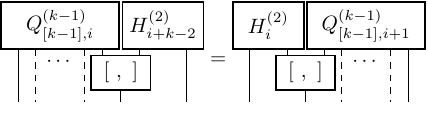}}
\end{equation}
for all $i\in\Lambda_N$,
which is equivalent to the condition $\hat{Q}_{[k-1]}^{(k-1)}\in\bound_{\leq}^{(k-1)}$. 
The uniqueness of $\hat{Q}_{[k-1]}^{(k-1)}$ follows from the injectivity~\eqref{eq:injective}. 
\end{proof}
If we define $\hat{Q}_{[k]}^{(k)}\in\bound^{(k)}$ as in Eq.~\eqref{eq:Step1_reduction} 
for some $\hat{Q}_{[k-1]}^{(k-1)}\in\bound_{\leq}^{(k-1)}$, 
then it is easy to verify $\hat{Q}_{[k]}^{(k)}\in\bound_{\leq}^{(k)}$. 
Therefore, 
we have established $\bound_{\leq}^{(k)}\cong\bound_{\leq}^{(k-1)}$ and Eq.~\eqref{eq:Step1}. 
By applying Eq.~\eqref{eq:Step1_reduction} recursively, 
$\hat{Q}^{(k)}_{[k],i}$ is represented as
\begin{equation}
\figurecentering{\includegraphics[]{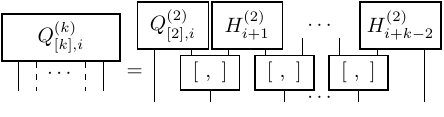}}
,
\label{eq:almost_dp}
\end{equation}
where $\hat{Q}_{[2]}^{(2)}=\sum_{i=1}^N\hat{Q}^{(2)}_{[2],i}$ is an element of $\bound^{(2)}_{\leq}$.

\subsection{Isomorphisms between $\bound^{(k)}_{<,b}$'s}
As pointed out at the end of Subsec.~\ref{subsec:main}, 
Thm.~\ref{thm:main} follows from Thm.~\ref{thm:general_step2}. 
Therefore, we conclude this section by proving Thm.~\ref{thm:general_step2}. 
Here, we shall prove only the case of $\bound_{<,b}^{(3)}\cong\bound_{<,b}^{(4)}$ 
to explain the bare essentials of the proof. 
A rigorous and complete proof of Eq.~\eqref{eq:Step2} is provided in Appx.~\ref{sec:complete_proof}
for interested readers. 

We introduce a new symbol to represent the commutator between $2$-local operators. 
We denote a linear map $[\bullet ,\hat{X}]:\bound_0^{\otimes 2}\to\bound^{\otimes 2}$ as 
\begin{equation}
\figurecentering{\includegraphics[]{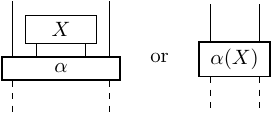}}
\ .
\end{equation}
If one line of the output is represented by a solid line instead of a dashed one, 
it indicates that a projection: 
\begin{equation}
    \bound\ni\hat{A}\mapsto \hat{A}-\frac{\Tr\hat{A}}{d}I\in\bound_0
\end{equation}
from $\bound$ to $\bound_0$ is performed at the corresponding site.

For $\alpha$ defined above, the following lemma holds, 
which is essential for the proof of Eq.~\eqref{eq:Step2}. 
\begin{lem}\label{lem:key}
For $\hat{C}\in\bound_0^{\otimes 2}$, the following identity holds:
\begin{equation}
\figurecentering{\includegraphics[scale=0.95]{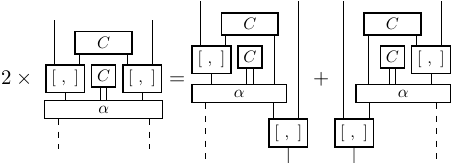}}
\ .
\label{eq:key}
\vspace{0mm}%
\end{equation}
\end{lem}
We provide a proof of this lemma in Appx.~\ref{sec:proof_lem}. 
By using this lemma, we can show the following proposition. 
\begin{prop}\label{prop:step2_3}
There is a linear isomorphism 
\begin{equation}
    \bound_{<,b}^{(4)}\ni q\hat{Q}_{b}^{(4)}+\hat{Q}_{[4]}^{(3)}\mapsto q\hat{Q}_{b}^{(3)}+\hat{Q}_{[3]}^{(2)}\in\bound_{<,b}^{(3)},
\end{equation}
where $q\in\mathbb{C}$ and 
$\hat{Q}_{b}^{(k)}$ is defined by Eq.~\eqref{eq:boosted_operator}.
\end{prop}
\begin{proof}
First we construct a map from $\bound_{<,b}^{(4)}$ to $\bound_{<,b}^{(3)}$. 
For $\hat{Q}_{[4]}^{(3)}\in\bound^{(3)}$,
the condition $q\hat{Q}_{b}^{(4)}+\hat{Q}_{[4]}^{(3)}\in\bound_{<,b}^{(4)}$ is equivalent to the following equalities:
\begin{widetext}
\begin{align}
\figurecentering{\includegraphics[scale=0.9]{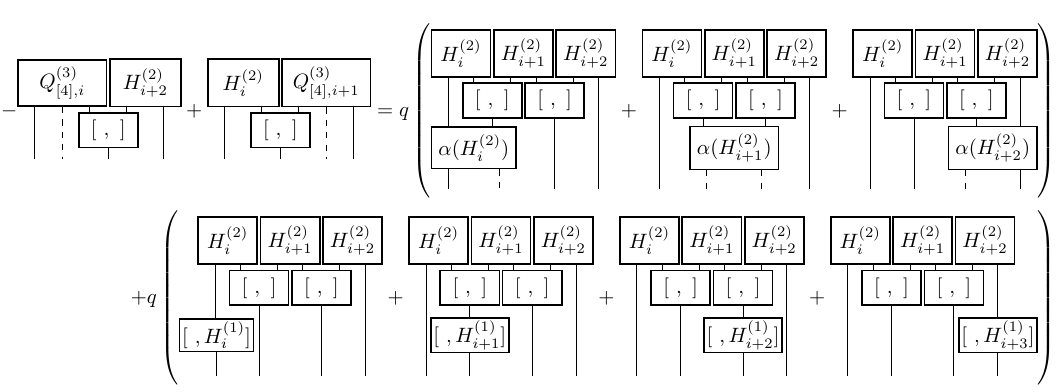}}
\label{eq:Step2_4}
\end{align}
\end{widetext}
for all $i\in\Lambda_N$. 
Here, $[\ ,\hat{X}]:\bound_0\ni\hat{Y}\mapsto[\hat{Y},\hat{X}]\in\bound_0$ is the commutator with $\hat{X}\in\bound_0$. 
Let us focus on the terms in the first bracket on the right-hand side (RHS) of Eq.~\eqref{eq:Step2_4}. 
By applying Lem.~\ref{lem:key} to the middle term, 
we can decompose these terms into the following form:
\begin{equation}
\figurecentering{\includegraphics[scale=1]{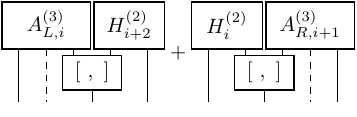}}
,
\label{eq:decomposition_two_com}
\end{equation}
where $\hat{A}^{(3)}_L$ and $\hat{A}^{(3)}_R$ satisfy
\begin{equation}
\hat{A}^{(3)}_{L,i}+\hat{A}^{(3)}_{R,i}
=  
\figurecentering{\includegraphics[scale=1]{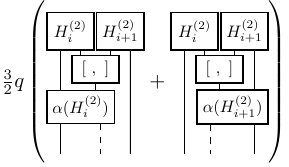}}
.
\label{eq:step2_bigcom}
\end{equation}

Similarly, by using the Jacobi identity, 
the terms in the second bracket on the RHS of Eq.~\eqref{eq:Step2_4} is decomposed into a similar form as Eq.~\eqref{eq:decomposition_two_com}~\footnote{
See Appx.~\ref{sec:complete_proof} for details. 
},
by two operators $\hat{B}^{(3)}_L$ and $\hat{B}^{(3)}_R$ satisfying 
\begin{align}
&\hat{B}^{(3)}_{L,i}+\hat{B}^{(3)}_{R,i}\nonumber\\
&=  
\figurecentering{\includegraphics[scale=0.9]{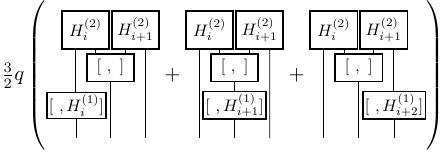}}
.
\label{eq:step2_scom}
\end{align}
The above relations allow us to simplify Eq.~\eqref{eq:Step2_4} as 
\begin{equation}
\figurecentering{\includegraphics[scale=1]{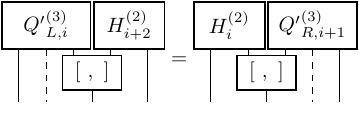}}
,
\label{eq:Step2_4_simplified}
\end{equation}
where two operators $\hat{Q'}^{(3)}_{L,i}$ and $\hat{Q'}^{(3)}_{R,i}$ are defined as
\begin{align}
    \hat{Q'}^{(3)}_{L,i}&\coloneqq\hat{Q}^{(3)}_{[4],i}+\hat{A}^{(3)}_{L,i}+\hat{B}^{(3)}_{L,i}
    ,\label{eq:step2_op_left_3}\\
    \hat{Q'}^{(3)}_{R,i}&\coloneqq\hat{Q}^{(3)}_{[4],i}-\hat{A}^{(3)}_{R,i}-\hat{B}^{(3)}_{R,i}
    .\label{eq:step2_op_right_3}
\end{align}
Similarly to Prop.~\ref{prop:step1},
there is an operator $\hat{Q'}^{(2)}_{i+1}\in\bound_0^{\otimes 2}$ satisfying
\begin{align}
&\figurecentering{\includegraphics[scale=1]{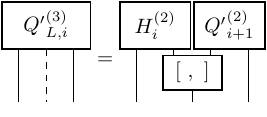}}
,
\label{eq:step2_left}\\
&\figurecentering{\includegraphics[scale=1]{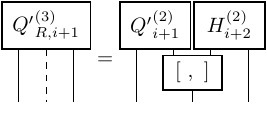}}
.
\label{eq:step2_right}
\end{align}
Equations~\eqref{eq:step2_op_left_3} and~\eqref{eq:step2_op_right_3} yield the following equality:
\begin{equation}
\hat{Q'}^{(3)}_{L,i}-\hat{Q'}^{(3)}_{R,i}
    =\hat{A}^{(3)}_{L,i}+\hat{B}^{(3)}_{L,i}+\hat{A}^{(3)}_{R,i}+\hat{B}^{(3)}_{R,i}.\label{eq:step2_reduced_eq}
\end{equation}
By substituting Eqs.~\eqref{eq:step2_bigcom},~\eqref{eq:step2_scom},~\eqref{eq:step2_left} and~\eqref{eq:step2_right} 
into Eq.~\eqref{eq:step2_reduced_eq}, 
we obtain
\begin{widetext}
\begin{equation}
\figurecentering{\includegraphics[scale=0.95]{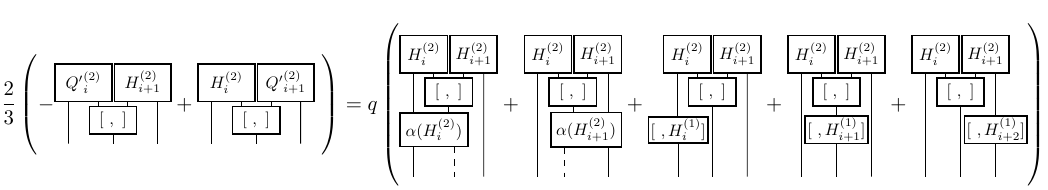}}
\label{eq:Step2_3}
\end{equation}
\end{widetext}
for all $i\in\Lambda_N$, 
which is equivalent to $q\hat{Q}_{b}^{(3)}+\hat{Q}_{[3]}^{(2)}\in\bound_{<,b}^{(3)}$, 
where we define $\hat{Q}_{[3]}^{(2)}$ as $\frac{2}{3}\sum_{i=1}^N\hat{Q'}^{(2)}_{i}$. 
Therefore, we have constructed a linear transformation $i_{34}:\bound_{<,b}^{(4)}\to\bound_{<,b}^{(3)}$. 
The injectivity of this map is confirmed from the injectivity~\eqref{eq:injective}. 
Conversely, by using the above equalities (e.g., Eq.~\eqref{eq:step2_op_left_3}), 
we can construct a mapping $i_{43}:\bound_{<}^{(3)}\to\bound_{<}^{(4)}$ satisfying $i_{34}\circ i_{43}=\id_{\bound_{<}^{(3)}}$. 
This implies the surjectivity of $i_{34}$. 
Furthermore, the action of $i_{43}$ on $\bound_{\leq}^{(2)}$ is the same as that of $\frac{2}{3}\iota_2$, which completes the proof. 
\end{proof}

\subsection{Absence of spectrum generating algebras}\label{subsec:SGA}
As discussed in Ref.~\cite{shiraishi2dimXY2024}, 
the proof for the absence of local conserved quantities can be slightly modified to apply to the analysis of a related algebraic structure: 
spectrum generating algebras. 
A system with a Hamiltonian $\hat{H}$ is said to admit a spectrum generating algebra (SGA) for an operator $\hat{Q}$
    if 
    \begin{equation}
        [\hat{Q},\hat{H}]=\mathcal{E}\hat{Q}\label{eq:SGA}
    \end{equation}
holds for some $\mathcal{E}\in\mathbb{R}\setminus\{0\}$. 
This algebraic structure is a variant of the condition $[\hat{Q},\hat{H}]=0$
for the existence of a conserved quantity, 
and $\eta$-pairing~\cite{yangEtaPairingOffdiagonal1989,yangSo4SymmetryHubbard1990} in the Hubbard model is a well-known example. 

In parallel with Thm.~\ref{thm:main}, 
the following theorem holds.
\begin{thm}\label{thm:SGA}
Let us consider a Hamiltonian $\hat{H}$ of the same form as in Thm.~\ref{thm:main}. 
If $\len([\hat{Q},\hat{H}]-\mathcal{E}\hat{Q})\leq 2$ does not hold
for any $\mathcal{E}\in\mathbb{R}\setminus\{0\}$ 
and any $3$-local quantity $\hat{Q}$,
then the system has no $k$-local quantity satisfying Eq.~\eqref{eq:SGA} for some $\mathcal{E}\in\mathbb{R}\setminus\{0\}$ and $3\leq k\leq N/2$. 
In other words, 
the system does not admit an SGA for a $k$-local quantity. 
\end{thm}
\begin{proof}
As a generalization of $\bound_<^{(k)}$, 
we introduce the following set of operators:
\begin{equation}
      \bound_{<,\mathcal{E}}^{(k)}\coloneqq
    \{\hat{X}\in\bound^{(k)}\oplus\bound^{(k-1)}\mid \len([\hat{X},\hat{H}]-\mathcal{E}\hat{X})\leq k-1\}.   
\end{equation}
Note that $\bound_{<,\mathcal{E}}^{(k)}\subset\bound_{\leq}^{(k)}\oplus\bound^{(k-1)}
=\mathbb{C}\hat{Q}_b^{(k)}\oplus\bound^{(k-1)}$ holds. 
To prove Thm.~\ref{thm:SGA} for the case of $k=4$, 
we first construct a linear transformation ${i'}_{34}:\bound_{<,\mathcal{E}}^{(4)}\to\bound_{<,\frac{2}{3}\mathcal{E}}^{(3)}$. 

In fact, this map can be defined in a manner similar to 
${i}_{34}:\bound_{<}^{(4)}\to\bound_{<}^{(3)}$. 
For $q\in\mathbb{C}$ and $\hat{Q}_{[4]}^{(3)}\in\bound^{(3)}$,
the condition $q\hat{Q}_b^{(4)}+\hat{Q}_{[4]}^{(3)}\in\bound_{<,\mathcal{E}}^{(4)}$ 
is equivalent to the following equalities:
\begin{align}
\figurecentering{\includegraphics[scale=0.95]{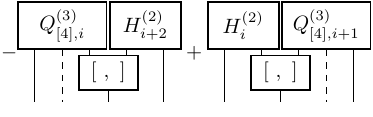}}
  &\nonumber\\
=\mbox{(the RHS of Eq.~\eqref{eq:Step2_4})}
&-\mathcal{E}q\hat{Q}_{b,i}^{(4)}
\label{eq:Step2_SGA_simplified}
\end{align}
for all $i\in\Lambda_N$. 
As in the proof of the previous subsection, 
Eq.~\eqref{eq:Step2_SGA_simplified} is equivalent to Eq.~\eqref{eq:Step2_4_simplified}, 
where $\hat{Q'}^{(3)}_{L,i}$ is redefined as 
\begin{equation}
\hat{Q'}^{(3)}_{L,i}\coloneqq\hat{Q}^{(3)}_{[4],i}+\hat{A}^{(3)}_{L,i}+\hat{B}^{(3)}_{L,i}-\mathcal{E}q\hat{Q}_{b,i}^{(3)}.%
\end{equation}
Under this definition, 
we can construct $\hat{Q}_{[3]}^{(2)}\in\bound^{(2)}$ satisfying
\begin{align}
\figurecentering{\includegraphics[scale=0.95]{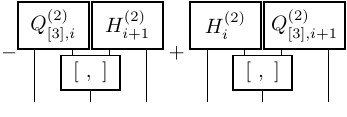}}
    &\nonumber\\
    =\mbox{(the RHS of Eq.~\eqref{eq:Step2_3})}
&-\frac{2}{3}\mathcal{E}q\hat{Q}_{b,i}^{(3)}
\end{align}
for all $i\in\Lambda_N$, 
which is equivalent to $q\hat{Q}_{b}^{(3)}+\hat{Q}_{[3]}^{(2)}\in\bound_{<,\frac{2}{3}\mathcal{E}}^{(3)}$. 

On the other hand, 
for any $\mathcal{E}\neq0$, 
$\hat{Q}_{[3]}^{(3)}+\hat{Q}_{[3]}^{(2)}\in\bound_{<,\frac{2}{3}\mathcal{E}}^{(3)}$ 
implies $q=0$ by the assumption. 
Therefore, 
we have $\bound_{<,\mathcal{E}}^{(4)}\subset\bound^{(3)}$ for every $\mathcal{E}\in\mathbb{R}\setminus\{0\}$, 
which indicates that 
there is no $4$-local quantity satisfying Eq.~\eqref{eq:SGA}. 
With a slight modification of the proof of Thm.~\ref{thm:general_step2} 
for the general $k$ case (Appx.~\ref{sec:complete_proof}), 
we can show the general $k$ case as well. 
\end{proof}

We note that SGAs are of interest in the study of thermalization. 
If an operator $\hat{Q}$ satisfies Eq.~\eqref{eq:SGA} on some subspaces, 
the system is said to admit a restricted spectrum generating algebra~\cite{markUnifiedStructureExact2020,moudgalyaEtapairingHubbardModels2020}. 
This structure offers a unified origin of the tower of quantum many-body scar states,
which prevents thermalization.

\section{Algorithm of integrability test}\label{sec:algorithm}
Using Thm.~\ref{thm:main} and Thm.~\ref{thm:general_step1}, 
we can construct an algorithm to test 
the integrability and nonintegrability of one-dimensional quantum spin systems. 
In this section, we only consider translationally invariant systems. 
In this case, the set of $k$-local operators can be decomposed into momentum sectors via the Fourier transformation as 
\begin{equation}
\bound^{(k)}=\bigoplus_{p\in\bigl\{\frac{2\pi}{N}m\bigm| m=0,\dots,N-1\bigr\}}\bound^{(k,p)},    
\end{equation}
where $p$ denotes the momentum, 
which corresponds to the eigenvalue of the translation operation $\tau:\bound^{\otimes A}\to\bound^{\otimes (A-1)}$. 
Since this decomposition is compatible with the subspaces $\bound^{(k)}_{\leq}$ and $\bound^{(k)}_{<}$, 
the integrability test introduced in Sec.~\ref{subsec:main} can be reduced to the following algorithm. 
First, we determine the subspace $\bound^{(2,p)}_{\leq}\coloneqq\bound^{(2)}_{\leq}\cap\bound^{(2,p)}$, 
which is equivalent to solving the following generalized eigenvalue problem:
\begin{equation}
\figurecentering{\includegraphics[scale=1]{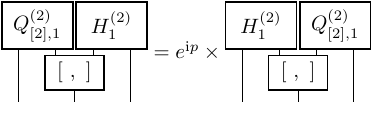}}
\label{eq:algorithm_k=2}
\end{equation}
for $\hat{Q}_{[2],1}^{(2)}\in\bound_0^{\otimes 2}$, 
where we use the translational invariance $\tau(H^{(2)}_{i+1})=H^{(2)}_i$ 
and the fact that $\hat{Q}_{[2]}^{(2)}\in\bound^{(2,p)}$ satisfies $\tau(\hat{Q}^{(2)}_{[2],i+1})=e^{\im p}\hat{Q}^{(2)}_{[2],i}$. 
We note that $\bound^{(2,p)}_{\leq}=\{0\}$ except for at most $(d^2-1)^2$ numbers of ($N$-independent) momentum $p$ %
due to the injectivity~\eqref{eq:injective}.
If $\bound_{\leq}^{(2,p=0)}=\mathbb{C}\hat{H}^{(2)}$, 
we can move on to the next step.

Next, we check whether $\bound^{(3)}_{<}$ in the zero-momentum sector is trivial (i.e., equal to $\bound^{(2,p=0)}_{\leq}$) or not. 
This step is equivalent to finding solutions of the following equation:
\begin{widetext}
\begin{equation}
\figurecentering{\includegraphics[scale=1]{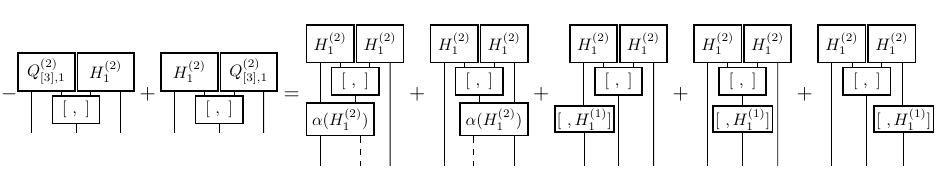}}
\label{eq:algorithm_k=3}
    \end{equation}
\end{widetext}
for $\hat{Q}_{[3],1}^{(2)}\in\bound_0^{\otimes 2}$, 
where we assume that $\bound_{\leq}^{(2,p=0)}=\mathbb{C}\hat{H}^{(2)}$. 
If Eq.~\eqref{eq:algorithm_k=3} has no solution, 
we can conclude that this system 
has no $k$-local and translationally invariant conserved quantity 
for $3\leq k\leq N/2$. 
Furthermore, if $\bound^{(2,p)}_{\leq}=\{0\}$ holds, or equivalently, 
Eq.~\eqref{eq:algorithm_k=2} has no nontrivial solution for $p\neq0$, 
then this system is nonintegrable. 
The above procedure involves eigenvalue problems for a finite number of linear maps, 
all of which are independent of system size.
As pointed out below Thm.~\ref{thm:main}, 
this algorithm can also be applied to a system of spin ladders and models with finite-range interactions.

Note that if additional structural information about the Hamiltonian is available, a more efficient algorithm may be constructed.
For example, suppose that there exists a subspace $\mathcal{A}$ of $\bound_0$ satisfying $\hat{H}_1^{(2)}\in\mathcal{A}^{\otimes 2}$. 
Then, by the injectivity and Eq.~\eqref{eq:algorithm_k=2}, we have $\hat{Q}_{[2],1}^{(2)}\in\mathcal{A}^{\otimes 2}$. 
Furthermore, 
a solution of Eq.~\eqref{eq:algorithm_k=3} satisfies $\hat{Q}_{[3],1}^{(2)}\in\mathcal{A}_2^{\otimes 2}$, 
where $\mathcal{A}_2$ is defined as
\begin{equation}
    \mathcal{A}_2\coloneqq\Span\{\hat{A}\hat{B}-\Tr(\hat{A}\hat{B})I,[\hat{A},\hat{H}_1^{(1)}]
    \mid\hat{A},\hat{B}\in\mathcal{A}\}. \label{eq:basis_step2}
\end{equation}
This observation is useful to show the nonintegrability of general spin-$S$ systems. 
As a simple example, let us consider the spin-$S$ Heisenberg chain. 
In this case, $\mathcal{A}$ consists of a linear combination of spin operators $\hat{S}^{\alpha}$, 
and $\mathcal{A}_2$ is eight-dimensional regardless of $S$ if $S\geq1$. 
Therefore, we can show the nonintegrability of the XYZ chain for $S\geq1$ as well as 
that of the Heisenberg chain for $S=1$~\cite{parkProofNonintegrabilitySpin12024,hokkyoProofAbsenceLocal2025}, 
see Appx.~\ref{sec:general_XYZ}.

\section{Nonintegrability of systems with internal degree of freedom}\label{sec:internal}
As pointed out in Ref.~\cite{shiraishi2dimXY2024}, 
the proof of nonintegrability for systems on hypercubic lattices and ladders with nearest-neighbor interactions 
can be reduced to that for ladder systems. 
In this section, we demonstrate that 
the problem of determining $\bound^{(k)}_<$ for systems with on-site interactions among internal degrees of freedom 
can be reduced to that of a ladder system. 
The Hamiltonian of such a system takes the following form: 
\begin{equation}
     \hat{H}=\sum_{(\alpha,i)\in\mathcal{S}\times\Lambda_N}
    \hat{H}_{\alpha;i}^{(2)}
    +\sum_{i\in\Lambda_N}\sum_{\mathcal{S}_0\substack{\subset\mathcal{S}\\\neq\emptyset}}
    \hat{H}_{\mathcal{S}_0;i}^{(1)}
    ,\label{eq:internal_Hamiltonian}
\end{equation}
where $\hat{H}_{\alpha;i}^{(2)}\in\bound_0^{\{(\alpha,i),(\alpha,i+1)\}}$
and $\hat{H}_{\mathcal{S}_0;i}^{(1)}\in\bound_0^{\mathcal{S}_0\times\{i\}}$ hold. 
Here, $\mathcal{S}$ denotes the set of labels for internal degree of freedom with $2\leq|\mathcal{S}|<\infty$. 
A typical example of such a model is the spin ladder system, where $|\mathcal{S}| = 2$. 
Systems with nearest-neighbor interactions on the hypercubic lattice $(\Lambda_N)^D$ correspond to 
the case $\mathcal{S} = (\Lambda_N)^{D-1}$.  
In particular, systems with nearest-neighbor interactions on the honeycomb lattice can be regarded 
as a special case of that on the square lattice (see Fig.~\ref{fig:honeycomb})
and its Hamiltonian takes the form of Eq.~\eqref{eq:internal_Hamiltonian}. 
\begin{figure}[tbp]
    \centering
    \includegraphics[width=0.6\linewidth]{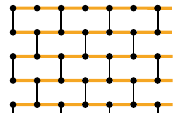}
    \caption{The honeycomb lattice can be regarded as a subgraph of the square lattice. 
    Therefore, a Hamiltonian with nearest-neighbor interactions on the honeycomb lattice takes the form of Eq.~\eqref{eq:internal_Hamiltonian}. 
    We define locality along the direction indicated by the horizontal lines.
    }\label{fig:honeycomb}
\end{figure}

We denote the interaction and the potential on $\{\alpha\}\times\Lambda_N$ as $\hat{H}_\alpha$:
\begin{equation}
    \hat{H}_{\alpha}=\sum_{i\in\Lambda_N}
    \hat{H}_{\alpha;i}^{(2)}
    +\sum_{i\in\Lambda_N}\hat{H}_{\{\alpha\};i}^{(1)}.
\end{equation}
We assume the injectivity~\eqref{eq:assumption_injective} of $\hat{H}_{\alpha}^{(2)}$ 
on the chain $\Lambda_N\times\{\alpha\}$
for all $\alpha\in\mathcal{S}$: for all $i\in\Lambda_N$, 
\begin{equation}
 \begin{aligned} 
  \hat{X}\in\bound_0^{\{(\alpha,i+1)\}},\ [I\otimes \hat{X},\hat{H}^{(2)}_{\alpha;i}]=0
  &\Rightarrow\hat{X}=0\\
  \mbox{and}\ \hat{Y}\in\bound_0^{\{(\alpha,i)\}},\ [\hat{Y}\otimes I,\hat{H}^{(2)}_{\alpha;i}]=0
  &\Rightarrow\hat{Y}=0,
  \end{aligned}\label{eq:indivisual_injective}
\end{equation}
which is equivalent to the injectivity of $\hat{H}^{(2)}$ on $\Lambda_N\times\mathcal{S}$.  

Notably, the Hamiltonian~\eqref{eq:internal_Hamiltonian} cannot satisfy the assumption~\eqref{eq:assumption_1dimensionality}.
Precisely speaking, the following proposition holds.
\begin{prop}\label{prop:step1_internal}
For the Hamiltonian of the form given in Eq.~\eqref{eq:internal_Hamiltonian} satisfying the injectivity (see Eq.~\eqref{eq:indivisual_injective}), 
the following holds:
\begin{equation}
    \bound^{(2)}_{\leq}(\hat{H})=\bigoplus_{\alpha\in\mathcal{S}}\bound^{(2)}_{\leq}(\hat{H}_\alpha).
\end{equation}
\end{prop}
This relation is clear from the fact that 
$\hat{H}^{(2)}$ is regarded as the Hamiltonian of independent $|\mathcal{S}|$ chains. 
Since $\dim\bound^{(2)}_{\leq}(\hat{H})=\sum_{\alpha}\dim\bound^{(2)}_{\leq}(\hat{H}_\alpha)\geq|\mathcal{S}|$ holds, 
the assumption~\eqref{eq:assumption_1dimensionality} does not hold 
and we cannot show the nonintegrability of these systems using Thm.~\ref{thm:main}. 
However, we can provide another test for the triviality of $\bound^{(k)}_{<}(\hat{H})$. 
By combining Prop.~\ref{prop:step1_internal} and Thm.~\ref{thm:general_step1}, 
we can express $\hat{Q}_{[k]}^{(k)}\in\bound^{(k)}_{\leq}(\hat{H})$ as
\begin{align}
    \hat{Q}_{[k]}^{(k)}
    &=
    \sum_{(\alpha,i)\in\mathcal{S}\times\Lambda_N}[\dots[[\hat{Q}_{[2],\alpha;i}^{(2)},\hat{H}_{\alpha;i+1}^{(2)}],\hat{H}_{\alpha;i+2}^{(2)}]\dots,\hat{H}_{\alpha;i+k-2}^{(2)}]\nonumber\\
    &\eqqcolon
    \sum_{\alpha\in\mathcal{S}}\hat{Q}_{[k],\alpha}^{(k)}\label{eq:step1_internal}
\end{align}
for some $\hat{Q}_{[2],\alpha}^{(2)}\in\bound^{(2)}_{\leq}(\hat{H}_\alpha)$. 
Then, we have the following theorem. 
\begin{thm}\label{thm:step2_internal}
    Consider the Hamiltonian of the form of Eq.~\eqref{eq:internal_Hamiltonian} satisfying the injectivity (see Eq.~\eqref{eq:indivisual_injective}). 
    For $3\leq k\leq N/2$ and $\hat{Q}_{[k]}^{(k)}\in\bound^{(k)}_{\leq}(\hat{H})$, 
    if $\hat{Q}_{[k]}^{(k)}+\hat{Q}_{[k]}^{(k-1)}\in\bound^{(k)}_{<}(\hat{H})$ holds for some $\hat{Q}_{[k]}^{(k-1)}\in\bound^{(k-1)}$, 
    then all of the following conditions hold. 
    \begin{enumerate}
        \item For each $\alpha\in\mathcal{S}$, 
        $\hat{Q}_{[k],\alpha}^{(k)}+\hat{Q}_{[k],\alpha}^{(k-1)}\in\bound^{(k)}_{<}(\hat{H}_\alpha)$ holds for some $(k-1)$-local operator $\hat{Q}_{[k],\alpha}^{(k-1)}$ on $\{\alpha\}\times\Lambda_N$. 
        \item For each $(\alpha,i)\in\mathcal{S}\times\Lambda_N$ and $\mathcal{S}_0\supsetneq\{\alpha\}$, 
        the following holds:
        \begin{align}
            &[[\hat{Q}^{(2)}_{[2],\alpha;i},\hat{H}^{(1)}_{\mathcal{S}_0;i+1}],\hat{H}^{(2)}_{\alpha;i+1}]%
            +[\hat{H}^{(2)}_{\alpha;i},[\hat{Q}^{(2)}_{[2],\alpha;i+1},\hat{H}^{(1)}_{\mathcal{S}_0;i+1}]]\nonumber\\
            &\hspace{15mm}=-(k-2)[[\hat{Q}^{(2)}_{[2],\alpha;i},\hat{H}^{(2)}_{\alpha;i+1}],\hat{H}^{(1)}_{\mathcal{S}_0;i+1}].\label{eq:step2_branch}
        \end{align}
        \item For each $i\in\Lambda_N$, $\alpha\neq\beta\in\mathcal{S}$ and $\mathcal{S}_0\ni\alpha,\beta$, 
        the following holds:
        \begin{align}
            &[[\hat{Q}^{(2)}_{[2],\alpha;i},\hat{H}^{(1)}_{\mathcal{S}_0;i+1}],\hat{H}^{(2)}_{\beta;i+1}]\nonumber\\
            &=[[\hat{H}^{(2)}_{\alpha;i},\hat{H}^{(1)}_{\mathcal{S}_0;i+1}],\hat{Q}^{(2)}_{[2],\beta;i+1}].\label{eq:step2_intra}
        \end{align}
    \end{enumerate}
    Here, $\hat{Q}_{[k],\alpha}^{(k)}$ and $\hat{Q}_{[2],\alpha}^{(2)}$'s are defined through Eq.~\eqref{eq:step1_internal}. 
\end{thm}
The first condition indicates that 
this system is no more integrable than a one-dimensional system $\hat{H}_\alpha$.  
In particular, 
if $\bound^{(k)}_{<}(\hat{H}_\alpha) = \bound^{(k-1)}_{\leq}(\hat{H}_\alpha)$ holds for all $\alpha$, 
then $\hat{H}$ is nonintegrable.
By combining with previous studies on the nonintegrability, 
This allows us to newly establish the nonintegrability of some systems, 
such as the quantum compass model on the cubic lattice. 

The remaining two conditions are equations for 3-local quantities. 
The second involves a single chain with a ``branch'',  
while the third describes the relation between two chains.  
From the theory of generalized eigenvalue problems~\cite{TemplatesSolutionAlgebraic2000}, 
Eq.~\eqref{eq:step2_branch} either has solutions for all $k$, 
or only for finitely many (at most $\dim \bound^{(2)}_{\leq}(\hat{H}_{\alpha})$) $k$.  
In particular, 
when $\bound^{(2)}_{\leq}(\hat{H}_{\alpha}) = \mathbb{C}\hat{H}^{(2)}_{\alpha}$ holds, 
Eq.~\eqref{eq:step2_branch} is equivalent, for all $k$, to the following equation:
\begin{equation}
    q_\alpha[[\hat{H}^{(2)}_{\alpha;i},\hat{H}^{(2)}_{\alpha;i+1}],
    \hat{H}^{(1)}_{\mathcal{S}_0;i+1}]=0,
\end{equation}
where we assume that $\hat{Q}^{(2)}_{[2],\alpha}=q_\alpha\hat{H}^{(2)}_{\alpha}$. 
This condition leads to the nonintegrability of models where $H_a$ is integrable, such as the XY(Z) model on the honeycomb lattice. 
In summary, the nonintegrability of the Hamiltonian in Eq.~\eqref{eq:internal_Hamiltonian} reduces to a problem on at most two chains with a ``branch''.

We conclude this section with the proof of Thm.~\ref{thm:step2_internal}. 
\begin{proof}[Proof of Thm.~\ref{thm:step2_internal}]
We decompose $\hat{Q}_{[k]}^{(k-1)}$ as
\begin{equation}
    \hat{Q}_{[k]}^{(k-1)}
    =\sum_{\substack{D\subset\mathcal{S}\times\Lambda_N\\\size{D}=k-1
    }}
     \hat{Q}_{[k],D}^{(k-1)},
\end{equation}
where the support of $ \hat{Q}_{[k],D}^{(k-1)}$ is $D$. 
We define the following abbreviation for special forms of $D$ and the corresponding $\hat{Q}_{[k],D}^{(k-1)}$:
\begin{align}
    D_{\alpha;i}^{(k)}&\coloneqq\{\alpha\}\times\{i,\dots,i+k-1\},\\
    \hat{Q}_{[k],\alpha;i}^{(k-1)}
    &\coloneqq\hat{Q}_{[k],D_{\alpha;i}^{(k-1)}}^{(k-1)},\\
    D_{\alpha\mathcal{S}_0\beta;i}^{(k,l)}&
    \coloneqq D_{\alpha;i}^{(l)}
    \cup\mathcal{S}_0\times\{i+l\}\cup D_{\beta;i+l+1}^{(k-l-1)},\\
    \hat{Q}_{[k],\alpha\mathcal{S}_0\beta;i}^{(k-1,l)}
    &\coloneqq\hat{Q}_{[k], D_{\alpha\mathcal{S}_0\beta;i}^{(k-1,l)}}^{(k-1)}.
\end{align}
We set $D_{\alpha;i}^{(l)}=\emptyset$ for $l\leq0$. 

First, we show the first condition. 
By tracing out the commutator $[\hat{Q}_{[k]}^{(k)}+\hat{Q}_{[k]}^{(k-1)},\hat{H}]$ except for $\{\alpha\}\times\Lambda_N$, 
we obtain 
\begin{align}
    [\hat{Q}_{[k],\alpha}^{(k)}+\hat{Q}_{[k],\alpha}^{(k-1)},\hat{H}_\alpha]\in\tilde{\bound}^{(k-1)}\cap\bound^{\otimes(\{\alpha\}\times\Lambda_N)},\nonumber\\
    \mbox{where } \hat{Q}_{[k],\alpha}^{(k-1)}
    \coloneqq\sum_{i\in\Lambda_N}\hat{Q}_{[k],\alpha;i}^{(k-1)}.\label{eq:step2_1dim}
\end{align}
Here we use the assumption $\len([\hat{Q}_{[k]}^{(k)}+\hat{Q}_{[k]}^{(k-1)},\hat{H}])\leq k-1$. 
It is clear that Eq.~\eqref{eq:step2_1dim} is equivalent to $\hat{Q}_{[k],\alpha}^{(k)}+\hat{Q}_{[k],\alpha}^{(k-1)}\in\bound^{k}_<(\hat{H}_\alpha)$. 

Next,  we show Eq.~\eqref{eq:step2_branch} and Eq.~\eqref{eq:step2_intra}. 
These equations are summarized as 
\begin{align}
            &[[\hat{Q}^{(2)}_{[2],\alpha;i},\hat{H}^{(1)}_{\mathcal{S}_0;i+1}],\hat{H}^{(2)}_{\beta;i+1}]%
            +[\hat{H}^{(2)}_{\alpha;i},[\hat{Q}^{(2)}_{[2],\beta;i+1},\hat{H}^{(1)}_{\mathcal{S}_0;i+1}]]\nonumber\\
            &\hspace{13mm}=-\delta_{\alpha\beta}(k-2)[[\hat{Q}^{(2)}_{[2],\alpha;i},\hat{H}^{(2)}_{\alpha;i+1}],\hat{H}^{(1)}_{\mathcal{S}_0;i+1}].
            \label{eq:step2_branches}
\end{align}
By tracing out the commutator $[\hat{Q}_{[k]}^{(k)}+\hat{Q}_{[k]}^{(k-1)},\hat{H}]$ except for $D_{\alpha\mathcal{S}_0\beta;i}^{(k,l)}$ 
and focusing on $k$-local operators, 
we have
\begin{equation}
    [\hat{Q}_{[k],\alpha\mathcal{S}_0\beta;i}^{(k-1,l=0)},\hat{H}_{\beta;i+k-2}^{(2)}]
    +
    [\hat{Q}_{[k],\beta;i}^{(k)},\hat{H}_{\mathcal{S}_0;i}^{(1)}]
    =0
\end{equation}
for $l=0$, 
\begin{align}
    [\hat{Q}_{[k],\alpha\mathcal{S}_0\beta;i}^{(k-1,l)},\hat{H}_{\beta;i+k-2}^{(2)}]
    +
    [\hat{Q}_{[k],\alpha\mathcal{S}_0\beta;i+1}^{(k-1,l-1)},\hat{H}_{\alpha;i}^{(2)}]\nonumber\\
    +
    \delta_{\alpha\beta}
    [\hat{Q}_{[k],\alpha;i}^{(k)},\hat{H}_{\mathcal{S}_0;i+l}^{(1)}]
    =0
\end{align}
for $1\leq l\leq k-2$ and 
\begin{equation}
    [\hat{Q}_{[k],\alpha\mathcal{S}_0\beta;i+1}^{(k-1,l=k-2)},\hat{H}_{\alpha;i}^{(2)}]
    +
    [\hat{Q}_{[k],\alpha;i}^{(k)},\hat{H}_{\mathcal{S}_0;i+k-1}^{(1)}]
    =0
\end{equation}
for $l=k-1$. 
These equations are translated by the string diagram as
\begin{equation}
    \figurecentering{\includegraphics[scale=0.7]{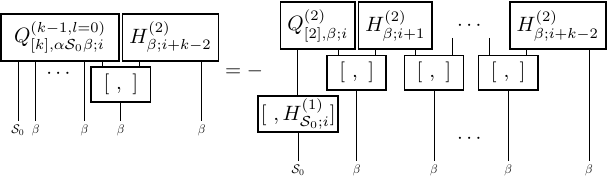}},\label{eq:step2_branch_right}
\end{equation}
\begin{equation}
    \figurecentering{\includegraphics[scale=0.7]{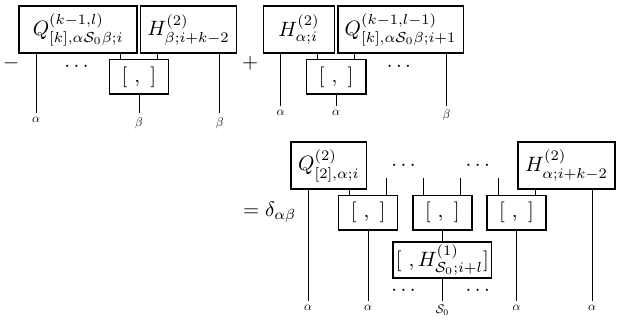}}\raisebox{-1cm}{,}\label{eq:step2_branch_mid}
\end{equation}
\begin{equation}
    \figurecentering{\includegraphics[scale=0.7]{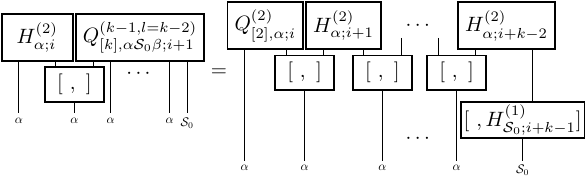}}
    \ .\label{eq:step2_branch_left}
\end{equation}
Here, the labels under the output lines represent the degrees of freedom on which the operators act.
By the injectivity, Eq.~\eqref{eq:step2_branch_left} yields
\begin{equation}
    \figurecentering{\includegraphics[scale=0.9]{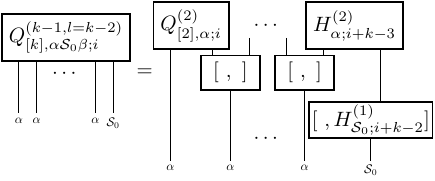}}\ . 
\end{equation}
By using Eq.~\eqref{eq:step2_branch_mid} iteratively, we have
\begin{equation}
    \figurecentering{\includegraphics[scale=0.7]{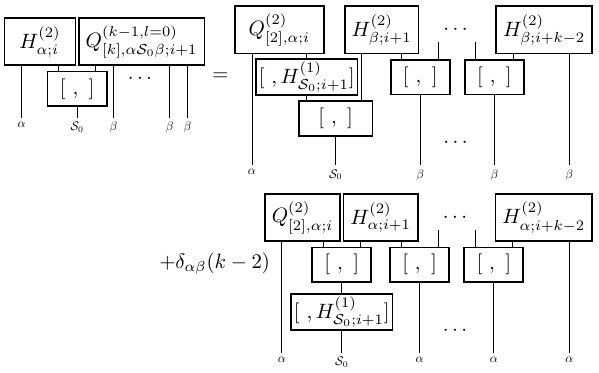}}\raisebox{-1.5cm}{.}
\end{equation}
This equation and Eq.~\eqref{eq:step2_branch_right} give
\begin{widetext}
    \begin{equation}
    \figurecentering{\includegraphics[scale=0.8]{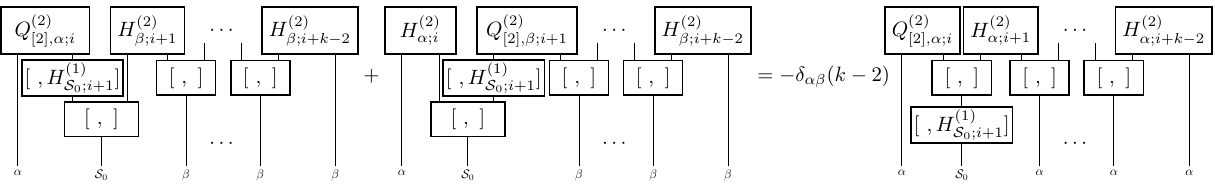}}\raisebox{-0cm}{,}
\end{equation}
\end{widetext}
which is equivalent to Eq.~\eqref{eq:step2_branches}. 
This concludes the proof.
\end{proof}

\section{Conclusion and outlook}\label{sec:summary}
In this work, we have proven a sufficient condition based solely on 3-local quantities that guarantees nonintegrability, 
i.e., the absence of local conserved quantities.
This result reveals a unified structure underlying existing proofs of nonintegrability, 
independent of the specific system, 
and provides a partial resolution to the long-standing conjecture that 
the existence of 3-local conserved quantities is a necessary condition for integrability. 
Using the obtained theorem, 
we have also constructed a procedure for determining 
whether a given translationally invariant quantum spin chain is integrable or nonintegrable, 
independent of system size. 
This opens the way for numerical determination of nonintegrability with a rigorous basis. 

One straightforward extension of this work to develop a unified understanding of nonintegrability in systems 
that do not satisfy the assumption~\eqref{eq:assumption_injective} or~\eqref{eq:assumption_1dimensionality}. 
Typical examples that do not satisfy the assumption~\eqref{eq:assumption_injective} are 
systems with Ising-type interactions, 
which are investigated in Refs.~\cite{chibaMixfield2024,chiba2dimIsing2024}. 
In these proofs of nonintegrability, 
it may be necessary to consider up to $(k-1)$-local outputs 
to show the absence of $k$-local conserved quantities~\cite{chibaMixfield2024}. 
Thus, the structure of nonintegrability may be qualitatively different from 
the nonintegrability of the system satisfying the assumption~\eqref{eq:assumption_injective}.
The XY model is an example that satisfies the assumption~\eqref{eq:assumption_injective} 
but fails to satisfy the assumption~\eqref{eq:assumption_1dimensionality}. 
Unlike the Ising-type interactions, 
the proof of the nonintegrability of this system~\cite{yamaguchiCompleteClassificationIntegrability2024,yamaguchiProofAbsenceLocal2024} is almost identical to that of systems satisfying both assumptions. 
Therefore, it is expected that our results can be extended to this case. 

Another promising direction is to explore the existence or absence of structures beyond integrability. 
In Sec.~\ref{subsec:XXZ}, 
we have demonstrated that even in the MBL phase, there is no (strictly) local conserved quantity. 
In the context of distinguishing the MBL phase from the ergodic phase where the ETH holds, 
a crucial question is whether the methods used to prove nonintegrability can be extended to quasi-local conserved quantities. 
As a related topic, 
we have discussed the absence of SGAs in Sec.~\ref{subsec:SGA}. 
If RSGAs could be analyzed, 
it may become possible to determine the presence or absence of quantum many-body scarring. 
These generalizations can lead to further classification of many-body systems 
beyond the presence or absence of local conserved quantities.

\begin{acknowledgments}
The author expresses sincere gratitude to Yuuya Chiba for carefully reviewing the manuscript and pointing out the gap in the proof. 
The author would like to thank Hosho Katsura, Mizuki Yamaguchi, and Naoto Shiraishi for fruitful discussions, 
and Masahito Ueda, Masaya Nakagawa, Hal Tasaki and Bal\'{a}zs Pozsgay
for their valuable comments on the manuscript.
The author is supported by Forefront Physics and Mathematics Program to Drive Transformation (FoPM), 
World-Leading INnovative Graduate Study (WINGS) Program, 
The University of Tokyo.
\end{acknowledgments}

\section*{Data availability}
No datasets were generated or analyzed during the current study.

\appendix
\section{Proof of Lemma~\ref{lem:key}}\label{sec:proof_lem}
We expand an operator $\hat{C}\in\bound_0^{\otimes 2}$ as $\hat{C}=\sum_a\hat{X}_a\otimes\hat{Y}_a$. 
Then, the RHS of Eq.~\eqref{eq:key} is expanded as 
\begin{align}
    &\hat{D}\otimes\hat{E}\nonumber\\
    \mapsto&
    \sum_{a,b}[[[\hat{D},\hat{X}_a]\otimes \hat{Y}_a,\hat{X}_b\otimes\hat{Y}_b],I\otimes\hat{E}]\nonumber\\
    &\ +[\hat{D}\otimes I,[\hat{X}_b\otimes [\hat{Y}_b,\hat{E}],\hat{X}_a\otimes\hat{Y}_a]]\nonumber\\
    =&\sum_{a,b}
    [\hat{D},\hat{X}_a]\hat{X}_b\otimes[\hat{Y}_a\hat{Y}_b,\hat{E}]
    -\hat{X}_b[\hat{D},\hat{X}_a]\otimes[\hat{Y}_b\hat{Y}_a,\hat{E}]\nonumber\\
    &\ +[\hat{D},\hat{X}_b\hat{X}_a]\otimes [\hat{Y}_b,\hat{E}]\hat{Y}_a
    -[\hat{D},\hat{X}_a\hat{X}_b]\otimes \hat{Y}_a[\hat{Y}_b,\hat{E}]\nonumber\\
    =&\sum_{a,b}
    [\hat{D},\hat{X}_a]\hat{X}_b\otimes[\hat{Y}_a,\hat{E}]\hat{Y}_b
    -\hat{X}_b[\hat{D},\hat{X}_a]\otimes\hat{Y}_b[\hat{Y}_a,\hat{E}]\nonumber\\
    &\ +[\hat{D},\hat{X}_b]\hat{X}_a\otimes [\hat{Y}_b,\hat{E}]\hat{Y}_a
    -\hat{X}_a[\hat{D},\hat{X}_b]\otimes \hat{Y}_a[\hat{Y}_b,\hat{E}]\nonumber\\
    =&2\sum_{a,b}
    [[\hat{D},\hat{X}_a]\otimes[\hat{Y}_a,\hat{E}],\hat{X}_b\otimes\hat{Y}_b],\label{eq:key_prf}
\end{align}
where we use the Leibniz rule of commutator: $[AB,C]=[A,C]B+A[B,C]$ in the second equality. 
The last line in Eq.~\eqref{eq:key_prf} corresponds to the left-hand side of Eq.~\eqref{eq:key}. $\square$

\section{Complete Proof of Theorem~\ref{thm:general_step2}}\label{sec:complete_proof}
 We provide complete the proof of Eq.~\eqref{eq:Step2}. 
 It is sufficient to show the following proposition. 
\begin{prop}\label{prop:step2_k}
Let $k\geq 4$. 
There is a linear isomorphism 
\begin{equation}
    \bound_{<,b}^{(k)}\ni q\hat{Q}_{b}^{(k)}+\hat{Q}_{[k]}^{(k-1)}\mapsto q\hat{Q}_{b}^{(k-1)}+\hat{Q}_{[k-1]}^{(k-2)}\in\bound_{<,b}^{(k-1)},
\end{equation}
where $q\in\mathbb{C}$ and 
$\hat{Q}_{b}^{(k)}$ is defined by Eq.~\eqref{eq:boosted_operator}.
\end{prop}
The proof of Prop.~\ref{prop:step2_k} follows the same procedure 
as for Prop.~\ref{prop:step2_3}. 
First we construct a map from $\bound_{<,b}^{(k)}$ to $\bound_{<,b}^{(k-1)}$. 
For $\hat{Q}_{[k]}^{(k-1)}\in\bound^{(k-1)}$,
the condition $q\hat{Q}_b^{(k)}+\hat{Q}_{[k]}^{(k-1)}\in\bound_{<,b}^{(k)}$ is equivalent to the following equalities:
\begin{widetext}
\begin{align}
\figurecentering{\includegraphics[scale=1]{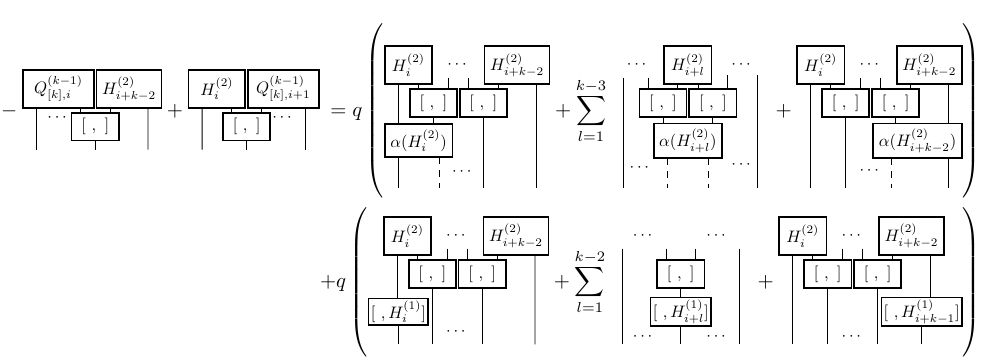}}
  \label{eq:Step2_k}
\end{align}
\end{widetext}
for all $i\in\Lambda_N$. 
We denote the terms in the first and second brackets on the RHS of Eq.~\eqref{eq:Step2_k} 
as $\hat{F}^{(k)}_i$ and $\hat{G}^{(k)}_i$, respectively. 
Similarly to the case of $k=4$, 
we can decompose $\hat{F}^{(k)}_i$ and $\hat{G}^{(k)}_i$ in terms of operators $\hat{A}^{(k-1)}_L,\hat{A}^{(k-1)}_R,\hat{B}^{(k-1)}_L$, and $\hat{B}^{(k-1)}_R$. 
Specifically, we have the following lemma for $\hat{F}^{(k)}_i$. 
\begin{lem}\label{lem:bigcom}
There exist operators $\hat{A}^{(k-1)}_L,\hat{A}^{(k-1)}_R\in\bound^{(k-1)}$ satisfying
\begin{align}
\hat{F}^{(k)}_i&=
\figurecentering{\includegraphics[scale=1]{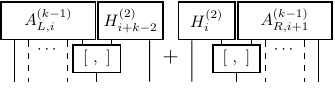}}
,
\label{eq:F_k}\\
\hat{F}^{(k-1)}_i&=
\frac{k-2}{k-1}
\left(\hat{A}^{(k-1)}_{L,i}+\hat{A}^{(k-1)}_{R,i}\right).\label{eq:F_k-1}
\end{align}
\end{lem}
Similarly, the following lemma holds for $\hat{G}^{(k)}_i$. 
\begin{lem}\label{lem:scom}
There exist operators $\hat{B}^{(k-1)}_L,\hat{B}^{(k-1)}_R\in\bound^{(k-1)}$ satisfying
\begin{align}
\hat{G}^{(k)}_i&=
\figurecentering{\includegraphics[scale=1]{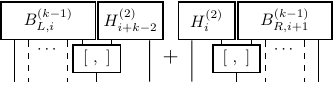}}
,
\label{eq:G_k}\\
\hat{G}^{(k-1)}_i&=
\frac{k-2}{k-1}
\left(\hat{B}^{(k-1)}_{L,i}+\hat{B}^{(k-1)}_{R,i}\right).\label{eq:G_k-1}
\end{align}
\end{lem}

By combining the above two lemmas, we can show Prop.~\ref{prop:step2_k}. 
\begin{proof}[Proof of Prop.~\ref{prop:step2_k}]
We define the following two operators:
\begin{align}
    \hat{Q'}^{(k-1)}_{L,i}&\coloneqq\hat{Q}^{(k-1)}_{[k],i}+q\hat{A}^{(k-1)}_{L,i}+q\hat{B}^{(k-1)}_{L,i},\label{eq:step2_op_left_k}\\
    \hat{Q'}^{(k-1)}_{R,i}&\coloneqq\hat{Q}^{(k-1)}_{[k],i}-q\hat{A}^{(k-1)}_{R,i}-q\hat{B}^{(k-1)}_{R,i}.\label{eq:step2_op_right_k}
\end{align}
By definition, we have 
\begin{equation}
    \hat{Q'}^{(k-1)}_{L,i}-\hat{Q'}^{(k-1)}_{R,i}=
    q\left(\hat{A}^{(k-1)}_{L,i}+\hat{A}^{(k-1)}_{R,i}
    +\hat{B}^{(k-1)}_{L,i}+\hat{B}^{(k-1)}_{R,i}\right). 
\end{equation}
Furthermore, by using Eq.~\eqref{eq:Step2_k}, 
we can find an operator $\hat{Q'}^{(k-2)}_{i+1}\in\bound^{(k-2)}$ satisfying
\begin{align}
\hspace{-5mm}
&\figurecentering{\includegraphics[scale=1]{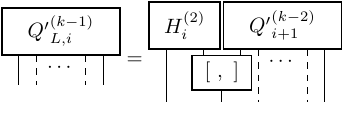}}
,\\
&\figurecentering{\includegraphics[scale=1]{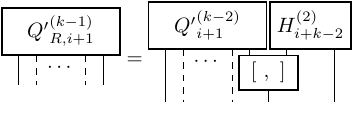}}
.
\end{align}
Therefore, for $\hat{Q}_{[k-1]}^{(k-2)}\coloneqq\frac{k-2}{k-1}\sum_{i}\hat{Q'}^{(k-2)}_{i}$ we have
\begin{equation}
\figurecentering{\includegraphics[scale=1]{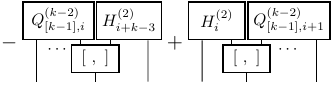}}
=q\left(\hat{F}^{(k-1)}_i+\hat{G}^{(k-1)}_i\right)
\end{equation}
for all $i\in\Lambda_N$, 
which is equivalent to $q\hat{Q}_{b}^{(k-1)}+\hat{Q}_{[k-1]}^{(k-2)}\in\bound_{<,b}^{(k-1)}$. 
A linear map from $\bound_{<,b}^{(k)}$ to $\bound_{<,b}^{(k-1)}$ is now constructed. 
By the same discussion as in the main text, this is in fact an isomorphism.
\end{proof}

The remaining task is to prove Lem.~\ref{lem:bigcom} and Lem.~\ref{lem:scom}. 
\begin{proof}[Proof of Lem.~\ref{lem:bigcom}]
The operator $\hat{F}^{(k)}_i$ consists of $k-1$ operators $\hat{F}^{(k)}_{i,0},\dots,\hat{F}^{(k)}_{i,k-2}$, 
where each $\hat{F}^{(k)}_{i,l}$ is defined as
\begin{equation}
    \hat{F}^{(k)}_{i,l}\coloneqq [\hat{Q}^{(k)}_{b,i},\hat{H}^{(2)}_{i+l}]. 
\end{equation}
By using Lem.~\ref{lem:key}, we can decompose $\hat{F}^{(k)}_{i,l}$ for $1\leq l\leq k-3$ as
\begin{align}
\hat{F}^{(k)}_{i,l}&=
\figurecentering{\includegraphics[scale=1]{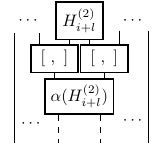}}
  \nonumber\\
&=
\figurecentering{\includegraphics[scale=1]{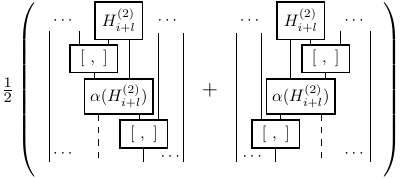}}
  \nonumber\\
&\eqqcolon\frac{1}{2}
(\hat{F}^{(k)}_{L,i,l}+\hat{F}^{(k)}_{R,i,l}). \label{eq:F_decomposition}
\end{align}
For $l=0$ and $l=k-2$, 
we define $\hat{F}^{(k)}_{L/R,i,l}$ as 
\begin{equation}
    \hat{F}^{(k)}_{L,i,l}=\hat{F}^{(k)}_{R,i,l}=\hat{F}^{(k)}_{i,l}. 
\end{equation}
Note that $\hat{F}^{(k)}_{L,i,l}$ and $\hat{F}^{(k)}_{R,i,l}$ satisfy the following equalities:
\begin{align}
    \hat{F}^{(k)}_{L,i,l}&
=\figurecentering{\includegraphics[scale=1]{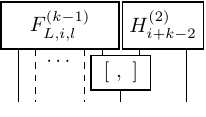}}
\ (0\leq l\leq k-3),\label{eq:F_left_to_left}\\
\hat{F}^{(k)}_{R,i,l}
&=\figurecentering{\includegraphics[scale=1]{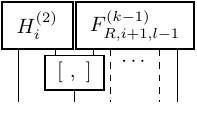}}
\ (1\leq l\leq k-2).\label{eq:F_right_to_right}
\end{align}
In addition, for $l\geq2$ we have 
\begin{equation}
     \hat{F}^{(k)}_{L,i,l}=
\hspace{-0mm}
\figurecentering{\includegraphics[scale=1]{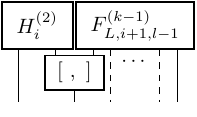}}
.
  \label{eq:F_left_to_right}
\end{equation}
Similarly, for $l\leq k-4$ we have 
\begin{equation}
\hat{F}^{(k)}_{R,i,l}=
\figurecentering{\includegraphics[scale=1]{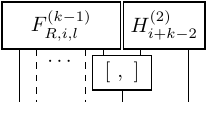}}
.
\label{eq:F_right_to_left}
\end{equation}

We now construct operators $\hat{A}^{(k-1)}_L,\hat{A}^{(k-1)}_R\in\bound^{(k-1)}$ as
\begin{align}
    \hat{A}^{(k-1)}_{L,i}&\coloneqq \frac{1}{k-2}\sum_{l=0}^{k-3}\left(\frac{l+2}{2}\hat{F}^{(k-1)}_{L,i,k-3-l}+\frac{l}{2}\hat{F}^{(k-1)}_{R,i,k-3-l}\right),\\
    \hat{A}^{(k-1)}_{R,i}&\coloneqq \frac{1}{k-2}\sum_{l=0}^{k-3}\left(\frac{l}{2}\hat{F}^{(k-1)}_{L,i,l}+\frac{l+2}{2}\hat{F}^{(k-1)}_{R,i,l}\right).
\end{align}
By using Eqs.~\eqref{eq:F_left_to_left}, \eqref{eq:F_right_to_right}, \eqref{eq:F_left_to_right}, and~\eqref{eq:F_right_to_left}, 
we can show that $\hat{A}^{(k-1)}_L$ and $\hat{A}^{(k-1)}_R$ satisfy Eq.~\eqref{eq:F_k}. 
Furthermore, Eq.~\eqref{eq:F_k-1} follows from Eq.~\eqref{eq:F_decomposition}, 
which completes this proof. 
\end{proof}

\begin{proof}[Proof of Lem.~\ref{lem:scom}]
The operator $\hat{G}^{(k)}_i$ consists of $k$ operators $\hat{G}^{(k)}_{i,0},\dots,\hat{G}^{(k)}_{i,k-1}$, 
where each $\hat{G}^{(k)}_{i,l}$ is defined as
\begin{equation}
    \hat{G}^{(k)}_{i,l}\coloneqq [\hat{Q}^{(k)}_{b,i},\hat{H}^{(1)}_{i+l}]. 
\end{equation}
We can decompose $\hat{G}^{(k)}_{i,l}$ for $1\leq l\leq k-2$ as
\begin{align}
\hat{G}^{(k)}_{i,l}&=
\figurecentering{\includegraphics[scale=1]{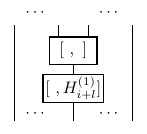}}
\nonumber\\
&=
\figurecentering{\includegraphics[scale=1]{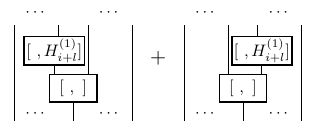}}
  \nonumber\\
&\eqqcolon
\hat{G}^{(k)}_{L,i,l}+\hat{G}^{(k)}_{R,i,l}, \label{eq:G_decomposition}
\end{align}
where we use the Jacobi identity: $[[\hat{A},\hat{B}],\hat{C}]=[[\hat{A},\hat{C}],\hat{B}]+[\hat{A},[\hat{B},\hat{C}]]$ in the second equality. 
For $l=0$ and $l=k-2$, 
we define $\hat{G}^{(k)}_{L/R,i,l}$ as 
\begin{align}
    \hat{G}^{(k)}_{R,i,0}&=\hat{G}^{(k)}_{i,0},\\
    \hat{G}^{(k)}_{L,i,k-1}&=\hat{G}^{(k)}_{i,k-1},\\
    \hat{G}^{(k)}_{L,i,0}&=\hat{G}^{(k)}_{R,i,k-1}=0.
\end{align}
Note that $\hat{G}^{(k)}_{L,i,l}$ and $\hat{G}^{(k)}_{R,i,l}$ satisfy the following equalities:
\begin{align}
    \hat{G}^{(k)}_{L,i,l}&=
\figurecentering{\includegraphics[scale=1]{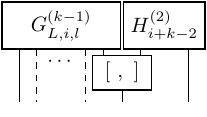}}
,\label{eq:G_left_to_left}\\
\hat{G}^{(k)}_{R,i,l}
&=
\figurecentering{\includegraphics[scale=1]{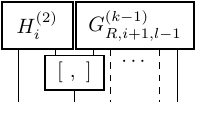}}
  \label{eq:G_right_to_right}
\end{align}
for $1\leq l\leq k-2$. 
In addition, for $l\geq2$ we have 
\begin{equation}
     \hat{G}^{(k)}_{L,i,l}=\figurecentering{\includegraphics[scale=1]{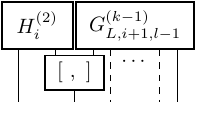}}
.
\label{eq:G_left_to_right}
\end{equation}
Similarly, for $l\leq k-3$ we have 
\begin{equation}
\hat{G}^{(k)}_{R,i,l}=\figurecentering{\includegraphics[scale=1]{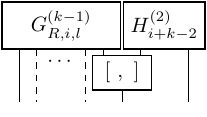}}
.
\label{eq:G_right_to_left}
\end{equation}

We now construct operators $\hat{B}^{(k-1)}_L,\hat{B}^{(k-1)}_R\in\bound^{(k-1)}$ as
\begin{align}
    \hat{B}^{(k-1)}_{L,i}&\coloneqq \frac{1}{k-2}\sum_{l=0}^{k-2}\left((l+1)\hat{G}^{(k-1)}_{L,i,k-2-l}+l\hat{G}^{(k-1)}_{R,i,k-2-l}\right),\\
    \hat{B}^{(k-1)}_{R,i}&\coloneqq \frac{1}{k-2}\sum_{l=0}^{k-2}\left(l\hat{G}^{(k-1)}_{L,i,l}+(l+1)\hat{G}^{(k-1)}_{R,i,l}\right).
\end{align}
By using Eqs.~\eqref{eq:G_left_to_left}, \eqref{eq:G_right_to_right}, \eqref{eq:G_left_to_right}, and~\eqref{eq:G_right_to_left}, 
we can show that $\hat{B}^{(k-1)}_L$ and $\hat{B}^{(k-1)}_R$ satisfy Eq.~\eqref{eq:G_k}. 
Furthermore, Eq.~\eqref{eq:G_k-1} follows from Eq.~\eqref{eq:G_decomposition}, 
which completes this proof. 
\end{proof}

\section{Nonintegrability of the XYZ chain with $S\geq1$}\label{sec:general_XYZ}
The Hamiltonian takes the following form:
\begin{equation}
    \hat{H}=\sum_{i=1}^N 
     \sum_{\alpha\in\{x,y,z\}}
     J_{\alpha}\hat{S}_i^\alpha\hat{S}_{i+1}^\alpha
    ,\label{eq:general_XYZ}
\end{equation}
where $\hat{S}^{\alpha}$'s represent the spin-$S$ operators for some $S\in\mathbb{Z}_{\geq2}/2$ 
and we assume that $J_x,J_y,J_z\neq0$. 
\begin{prop}
The spin-$S$ XYZ chain for $S\geq1$~\eqref{eq:general_XYZ} has no $k$-local conserved quantity for $3\leq k\leq N/2$. 
\end{prop}
\begin{proof}
As noted immediately after Eq.~\eqref{eq:assumption_injective}, 
this Hamiltonian satisfies the injectivity (see Eq.~\eqref{eq:assumption_injective}). 
In fact, if $[\hat{X}\otimes I,\hat{H}_i^{(2)}]$ holds, 
then $[\hat{X},\hat{S}_i^\alpha]$ holds for $\alpha=x,y,z$, 
which implies $\hat{X}=0$. 
By the argument at the last paragraph of Sec.~\ref{sec:algorithm}, 
the solution of Eq.~\eqref{eq:algorithm_k=2} satisfies $\hat{Q}_{[2],1}^{(2)}\in\mathcal{A}^{\otimes 2}$. 
Here, $\mathcal{A}$ is defined by
\begin{equation}
    \mathcal{A}\coloneqq\Span\{\hat{S}^\alpha\mid\alpha=x,y,z\}. 
\end{equation}
Thus, Eq.~\eqref{eq:algorithm_k=2} is equivalent for all $S$, 
and as in the case of $S=1/2$~\cite{shiraishiProofAbsenceLocal2019,yamaguchiCompleteClassificationIntegrability2024,yamaguchiProofAbsenceLocal2024}, 
its solutions are limited to $p=0$ and $\hat{Q}_{[2],1}^{(2)}\propto\hat{H}_{[2],1}^{(2)}$. 
We provide an explicit proof of this fact below.
First we expand $\hat{Q}_{[2],1}^{(2)}$ by the spin operators:
\begin{equation}
    \hat{Q}_{[2],1}^{(2)}=
     \sum_{\alpha\in\{x,y,z\}}
     q_{\alpha\beta}\hat{S}_1^\alpha\hat{S}_{2}^\beta.
\end{equation}
By using these $q_{\alpha\beta}$'s, both sides of Eq.~\eqref{eq:algorithm_k=2} are expanded as
\begin{align}
&\im\sum_{\alpha,\beta,\gamma,\delta}q_{\alpha\delta}J_{\gamma}\varepsilon_{\delta\gamma\beta}\hat{S}_1^\alpha\hat{S}_{2}^\beta\hat{S}_3^\gamma\nonumber\\
&=
\im e^{\im p}\sum_{\alpha,\beta,\gamma,\delta}J_{\alpha}q_{\delta\gamma}\varepsilon_{\alpha\delta\beta}\hat{S}_1^\alpha\hat{S}_{2}^\beta\hat{S}_3^\gamma. 
\end{align}
Hence, we have 
\begin{equation}
    \sum_{\delta}q_{\alpha\delta}J_{\gamma}\varepsilon_{\delta\gamma\beta}
    =e^{\im p}\sum_{\delta}J_{\alpha}q_{\delta\gamma}\varepsilon_{\alpha\delta\beta}
    \label{eq:algorithm_k=2_XYZ}
\end{equation}
for any $\alpha,\beta,\gamma\in\{x,y,z\}$. 
If we set $\beta=\alpha$, the RHS of Eq.~\eqref{eq:algorithm_k=2_XYZ} vanishes and we obtain 
$q_{\alpha\alpha'}=0$ for $\alpha\neq\alpha'$. 
By substituting this equality to Eq.~\eqref{eq:algorithm_k=2_XYZ}, we have
\begin{equation}
    q_{\alpha\alpha}J_\gamma=e^{\im p}J_\alpha q_{\gamma\gamma}
\end{equation}
for any $\alpha\neq\gamma\in\{x,y,z\}$. 
This implies $p=0$ and $q_{\alpha\alpha}/J_\alpha=q_{\gamma\gamma}/J_\gamma$, 
which is equivalent to $\hat{Q}_{[2],1}^{(2)}\propto\hat{H}_{[2],1}^{(2)}$. 

Next, we show that Eq.~\eqref{eq:algorithm_k=3} has no solution. 
By the argument at the last paragraph of Sec.~\ref{sec:algorithm}, 
the solution of Eq.~\eqref{eq:algorithm_k=3} satisfies $\hat{Q}_{[3],1}^{(2)}\in\mathcal{A}_2^{\otimes 2}$, 
where $\mathcal{A}_2$ is defined in Eq.~\eqref{eq:basis_step2}. 
For $S\geq1$, this space is eight-dimensional, whose basis consists of 
the following operators:
\begin{align}
    &\hat{S}^{\alpha\alpha}\coloneqq(\hat{S}^\alpha)^2-2I\ (\alpha=x,y),\nonumber\\
    &\hat{S}^{\alpha\beta}\coloneqq\{\hat{S}^\alpha,\hat{S}^\beta\}\ ((\alpha,\beta)=(x,y),(y,z),(z,x)),
\end{align}
and $\hat{S}^{\alpha=x,y,z}$. 
With respect to this basis, 
we focus on the coefficients of $\hat{S}^z_1 \hat{S}^z_2 \hat{S}^{xx}_3$ on both sides of Eq.~\eqref{eq:algorithm_k=3}. 
On the right-hand side, the only contribution to this term 
comes from the following single commutator:
\begin{align}
    &[\im J_zJ_x\hat{S}^{z}_1\hat{S}^{y}_2\hat{S}^{x}_3,J_x\hat{S}^{x}_2\hat{S}^{x}_3]\nonumber\\
    &= J_z(J_x)^2\hat{S}^z_1\hat{S}^z_2(\hat{S}^{xx}_3+2I).
\end{align}
On the left-hand side, the contributing commutator takes the following form:
\begin{align}
    &[J_z\hat{S}^{z}_1\hat{S}^{z}_2,\hat{S}^{\alpha\beta}_2\hat{S}^{xx}_3]\nonumber\\
    &=J_z\hat{S}^z_1[\hat{S}^z_2,\hat{S}^{\alpha\beta}_2]\hat{S}^{xx}_3
\end{align}
However, the commutator $[\hat{S}^z,\hat{S}^{\alpha\beta}]$ is calculated as 
\begin{align}
    [\hat{S}^z,\hat{S}^{xx}]&=\im\hat{S}^{xy},\\
    [\hat{S}^z,\hat{S}^{yy}]&=-\im\hat{S}^{xy},\\
    [\hat{S}^z,\hat{S}^{xy}]&=2\im(\hat{S}^{yy}-\hat{S}^{xx}),\\
    [\hat{S}^z,\hat{S}^{yz}]&=-\im\hat{S}^{zx},\\
    [\hat{S}^z,\hat{S}^{zx}]&=\im\hat{S}^{yz},
\end{align}
which indicates that none of these commutators yield $\hat{S}^z_1\hat{S}^z_2\hat{S}^{xx}_3$. 
Hence, Eq.~\eqref{eq:algorithm_k=3} leads to $J_z(J_x)^2 = 0$, which is a contradiction. 
Therefore, Eq.~\eqref{eq:algorithm_k=3} has no solution, 
and there are no nontrivial local conserved quantities of this system. 
\end{proof}

\bibliography{RigorousTest}
\end{document}